\documentclass[11pt,twocolumn,notitlepage]{article}
\usepackage[small,bf]{caption}
\usepackage{ownstyle}
\usepackage{tabularx}
\usepackage{authblk}
\usepackage{balance}

\usepackage{enumitem}
\setlist{nolistsep}

\title{Optimal Assembly for High Throughput Shotgun Sequencing\thanks{
The authors thank Yun Song, Lior Pachter, Sharon Aviran, and Serafim Batzoglou for stimulating discussions. 
This work is supported by the Center for Science of Information (CSoI), an NSF Science and Technology Center, under grant agreement CCF-0939370. M. Bresler is also supported by NSF grant DBI-0846015.}}
\author{Guy Bresler} \author{Ma'ayan Bresler} \author{David Tse\thanks{Author names are in alphabetical order.}}
\affil{Dept. of EECS, UC Berkeley. Email: \textsf{\{gbresler,mbresler,dtse\}@eecs.berkeley.edu}}
\begin{document}
\date{}
\onecolumn
\maketitle
\begin{abstract}
We present a framework for the design of optimal assembly algorithms for shotgun sequencing under the criterion of complete reconstruction. We derive a lower bound on the read length and the coverage depth required for reconstruction in terms of the repeat statistics of the genome. Building on earlier works, we design a de Brujin graph based assembly algorithm which can achieve very close to the lower bound for repeat statistics of a wide range of sequenced genomes, including the GAGE datasets. The results are based on a set of necessary and sufficient conditions on the DNA sequence and the reads for reconstruction. The conditions can be viewed as the shotgun sequencing analogue  of Ukkonen-Pevzner's necessary and sufficient conditions for Sequencing by Hybridization.
\end{abstract}
\thispagestyle{empty}

\clearpage

\twocolumn
\section{Introduction}

\setcounter{page}{1}
\label{sec:introduction}
\vspace{-1mm}

\subsection{Problem statement}

DNA sequencing is the basic workhorse of modern day biology and medicine. Since the sequencing of the Human Reference Genome ten years ago, there has been an explosive advance in sequencing technology, resulting in several orders of magnitude increase in throughput and decrease in cost. Multiple ``next-generation" sequencing platforms  have emerged.  All of them are based on the whole-genome shotgun sequencing method, which entails two steps. First, many short reads are extracted from random locations on the DNA sequence, with the length, number, and error rates of the reads depending on the particular sequencing platform. Second, the reads are assembled to reconstruct the original DNA sequence. 


Assembly of the reads is a major algorithmic challenge, and over the years dozens of assembly algorithms have been proposed to solve this problem \cite{wiki:assembly}. 
Nevertheless, the assembly problem is far from solved, and it is not clear how to compare algorithms nor where improvement might be possible. 
The difficulty of comparing algorithms is evidenced by the recent assembly evaluations Assemblathon 1 \cite{earl2011assemblathon} and GAGE \cite{Sal11}, where which assembler is ``best" depends on the particular dataset as well as the performance metric used.
In part this is a consequence of metrics for partial assemblies: there is an inherent tradeoff between larger contiguous fragments (contigs) and fewer mistakes in merging contigs (misjoins).
But more fundamentally,
independent of the metric, performance depends critically on the dataset, i.e. length, number, and quality of the reads, as well as the complexity of the genome sequence. 
With an eye towards the near future, we seek to understand the interplay between these factors by using the intuitive and unambiguous metric of \emph{complete} reconstruction\footnote{The notion of complete  reconstruction can be thought of as a mathematical idealization of the notion of ``finishing" a sequencing project as defined by the National Human Genome Research Institute \cite{finishing}, where finishing a chromosome requires at least 95\% of the chromosome to be represented by a contiguous sequence.}. Note that this objective of reconstructing the original DNA sequence from the reads
contrasts with the many \emph{optimization-based} formulations of assembly, such as shortest common superstring (SCS) \cite{KM93}, maximum-likelihood \cite{Mye95}, \cite{medvedev2009maximum}, and various graph-based formulations \cite{PTW01}, \cite{Mye05}. When solving one of these alternative formulations, there is no guarantee that the optimal solution is indeed the original sequence.

Given the goal of complete reconstruction,  the most basic questions are 1) {\bf feasibility}: given a set of reads, is it \emph{possible} to reconstruct the original sequence?  2) {\bf optimality}: which \emph{algorithms} can successfully reconstruct whenever it is feasible to reconstruct?  
The feasibility question is a measure of the intrinsic {\em information} each read provides about the DNA sequence, and for given sequence statistics depends on characteristics of the sequencing technology such as read length and noise statistics. As such, it can provide an algorithm-independent basis for evaluating the efficiency of a sequencing technology. Equally important, algorithms can be evaluated on their relative read length and data requirements, and compared against the fundamental limit.

In studying these questions, we consider the most basic shotgun sequencing model where $N$ noiseless reads\footnote{Reads are thus exact subsequences of the DNA.} of a fixed length $L$ base pairs are uniformly and independently drawn from a DNA sequence of length $G$. 
In this statistical model, feasibility is rephrased as the question of whether, for given sequence statistics, the correct sequence can be reconstructed with probability $1-\eps$ when $N$ reads of length $L$ are sampled from the genome.
We note that answering the feasibility question of whether each $N,L$ pair is sufficient to reconstruct is equivalent to finding the minimum required $N$ (or the \emph{coverage depth} $c=NL/G$) as a function of $L$.


A lower bound on the minimum coverage depth needed was obtained by Lander and Waterman \cite{LW88}. Their lower bound $\clw = \clw(L,\eps)$ is the minimum number of randomly located reads needed to cover the entire DNA sequence with a given target success probability $1-\eps$. 
While this is clearly a necessary condition, it is in general not tight: only requiring the reads to cover the entire genome sequence does not guarantee that consecutive reads can actually be stitched back together to recover the original sequence. 
Characterizing when the reads can be reliably stitched together, i.e. determining feasibility, is an open problem.  In fact, the ability to reconstruct depends crucially on the {\em repeat statistics} of the DNA sequence. 

An earlier work \cite{MBT12} has answered the feasibility and optimality questions under an i.i.d. model for the DNA sequence. However, real DNA, especially those of eukaryotes, have much longer and complex repeat structures. Here, we are interested in determining feasibility and optimality given {\em arbitrary} repeat statistics. This allows us to evaluate algorithms on statistics from already sequenced genomes, and gives confidence in predicting whether the algorithms will be useful for an \emph{unseen} genome with similar statistics.
 

\subsection{Results}

Our approach results in a pipeline, which takes as input a genome sequence and desired success probability $1-\eps$, computes a few simple repeat statistics, and from these statistics  computes a feasibility plot that indicates for which $L,N$ reconstruction is possible. Fig.~\ref{fig:repeatstatsINTRO} displays the simplest of the statistics, the number of repeats as a function of the repeat length $\ell$. Fig.~\ref{fig:plotINTRO} shows the resulting feasibility plot produced for the statistics of human chromosome 19 (henceforth \hc{19}) with success probability $99\%$. The horizontal axis signifies read length $L$ and the vertical axis signifies the normalized coverage depth $\bar c:=c/\clw$, the  
coverage depth $c$ normalized by $\clw$, the coverage depth required as per Lander-Waterman \cite{LW88} in order to cover the sequence.  

\begin{figure}[htb]
\centering{\includegraphics[height=1.6in]{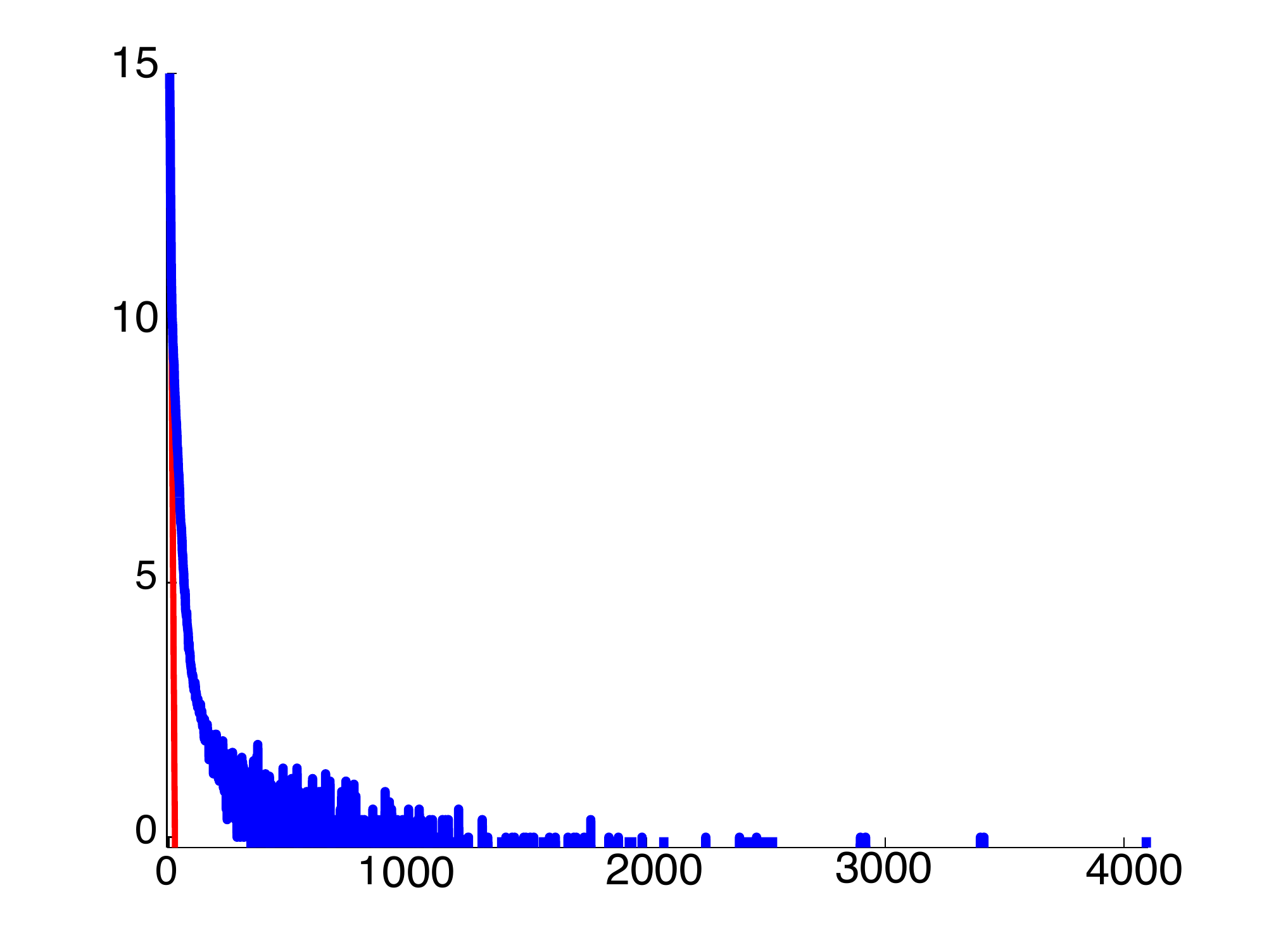}}
\caption{For \hc{19}, a log plot of number of repeats as a function of the repeat length $\ell$. Red line is what would have been predicted by an i.i.d. fit.}
\label{fig:repeatstatsINTRO}
\end{figure}

\begin{figure}
\centering{
\includegraphics[height=2.6in]{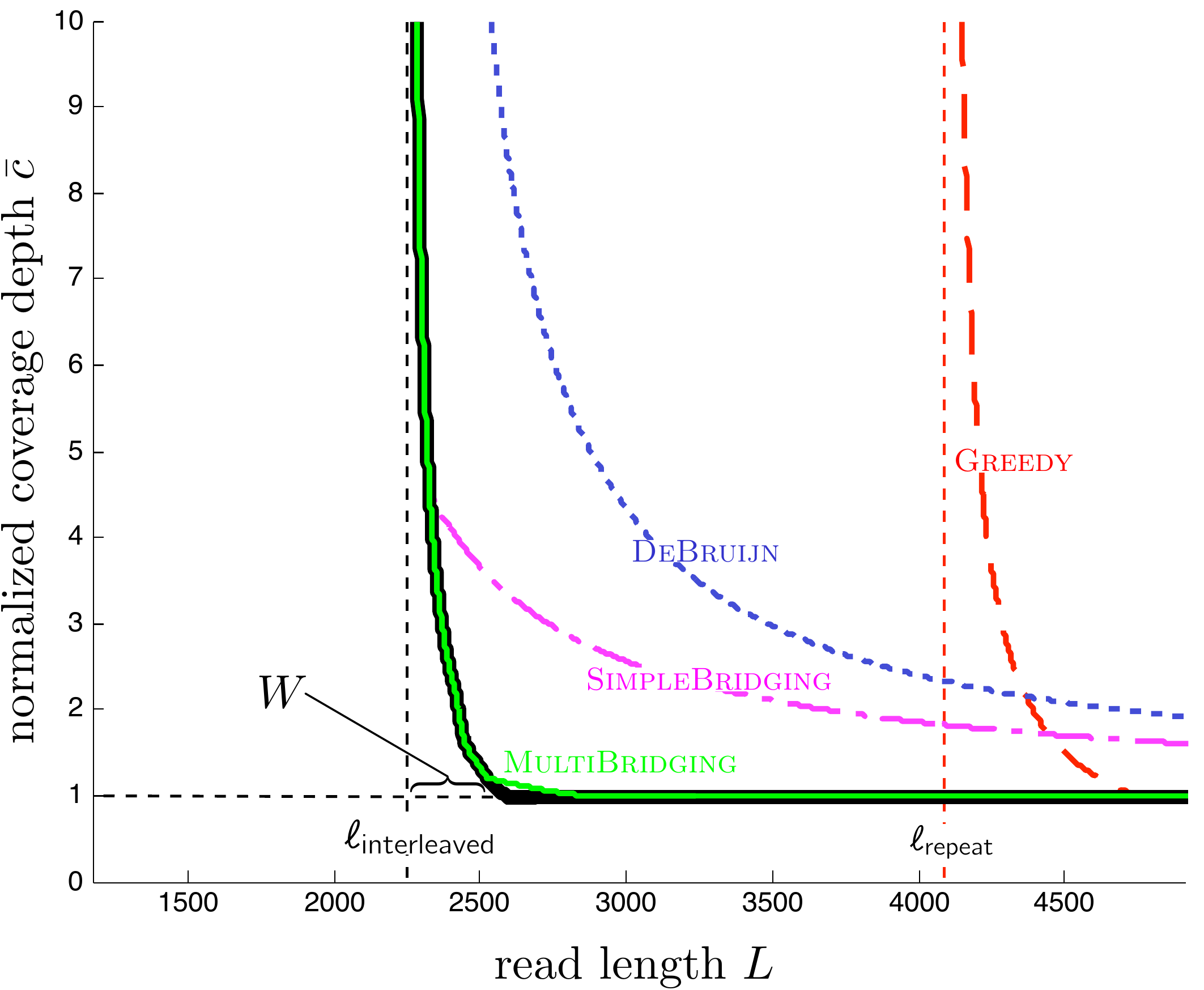}
\caption{Thick black lines are lower bounds on feasibility which holds for all algorithms, and colored curves are performance achieved by specific algorithms. Four such curves are shown: the greedy algorithm and three de Brujin graph based algorithms. }
\label{fig:plotINTRO}}
\vspace{-3mm}
\end{figure}

Since the coverage depth must satisfy $c \geq \clw$, the normalized coverage depth satisfies $\cb \geq 1$, and we plot the horizontal line $\cb = 1$. This lower bound holds for {\em any}  assembly algorithm. In addition,  there is another lower bound, shown as the thick black nearly vertical line in Fig.~\ref{fig:plotINTRO}. In contrast to the coverage lower bound, this lower bound is a function of the repeat statistics.  It has a vertical asymptote at $\lcrit := \max\{\Lint,\Ltri\}+1$, where $\Lint$ is the length of the longest interleaved repeats in the DNA sequence  and $\Ltri$  is the length of the longest triple repeat (see Section~\ref{sec:Lower} for precise definitions).   Our lower bound can be viewed as a generalization
of a result of Ukkonen \cite{Ukk92} for Sequencing by Hybridization to the shotgun sequencing setting. 

Each colored curve in the feasibility plot is the lower boundary of the set of feasible $N,L$ pairs for a specific algorithm. The rightmost curve  is the one achieved by the greedy algorithm, which merges reads with largest overlaps first  (used for example in TIGR \cite{Sut95}, CAP3~\cite{HM99}, and more recently SSAKE~\cite{War07}). As seen in Fig.~\ref{fig:plotINTRO}, its performance curve asymptotes at  $L=\Lrep$, the length of the longest repeat.  De Brujin graph based algorithms (e.g. \cite{IW95} and \cite{PTW01}) take a more global view via the construction of a de Brujin graph out of all the K-mers of the reads. The performance curves of all K-mer graph based algorithms asymptote at read length $L = \lcrit$, but different algorithms  use read information in a variety of ways to resolve repeats in the K-mer graph and thus have different coverage depth requirement beyond read length $\lcrit$.  By combining the ideas from several existing algorithms (including \cite{PTW01}, \cite{peng2010idba}) we designed {\bdbII}, which is very close to the lower bound for this dataset. Thus Fig.~\ref{fig:plotINTRO} answers, up to a very small gap, the feasibility of assembly for the repeat statistics of \hc{19}, where successful reconstruction is desired with probability $99\%$. 

We produce similar plots for a dozen or so datasets (see supplementary material). For datasets where $\Lint$ is significantly larger than $\Ltri$ (the majority of the datasets we looked at,  including those used in the recent GAGE assembly algorithm evaluation \cite{Sal11}), {\bdbII} is near optimal, thus  allowing us to characterize the fundamental limits for these repeat statistics (Fig. \ref{fig:sim}). On the other hand, if $\Ltri$ is close to or larger than $\Lint$, there is a gap between the performance of {\bdbII} and the lower bound (see for example Fig.~\ref{fig:bridging2triplesINTRO}). The reason for the gap is explained in Section~\ref{sec:gap}.

 \begin{figure}[htb]
\begin{centering}
\vspace{-2mm}
\includegraphics[height=1.6in]{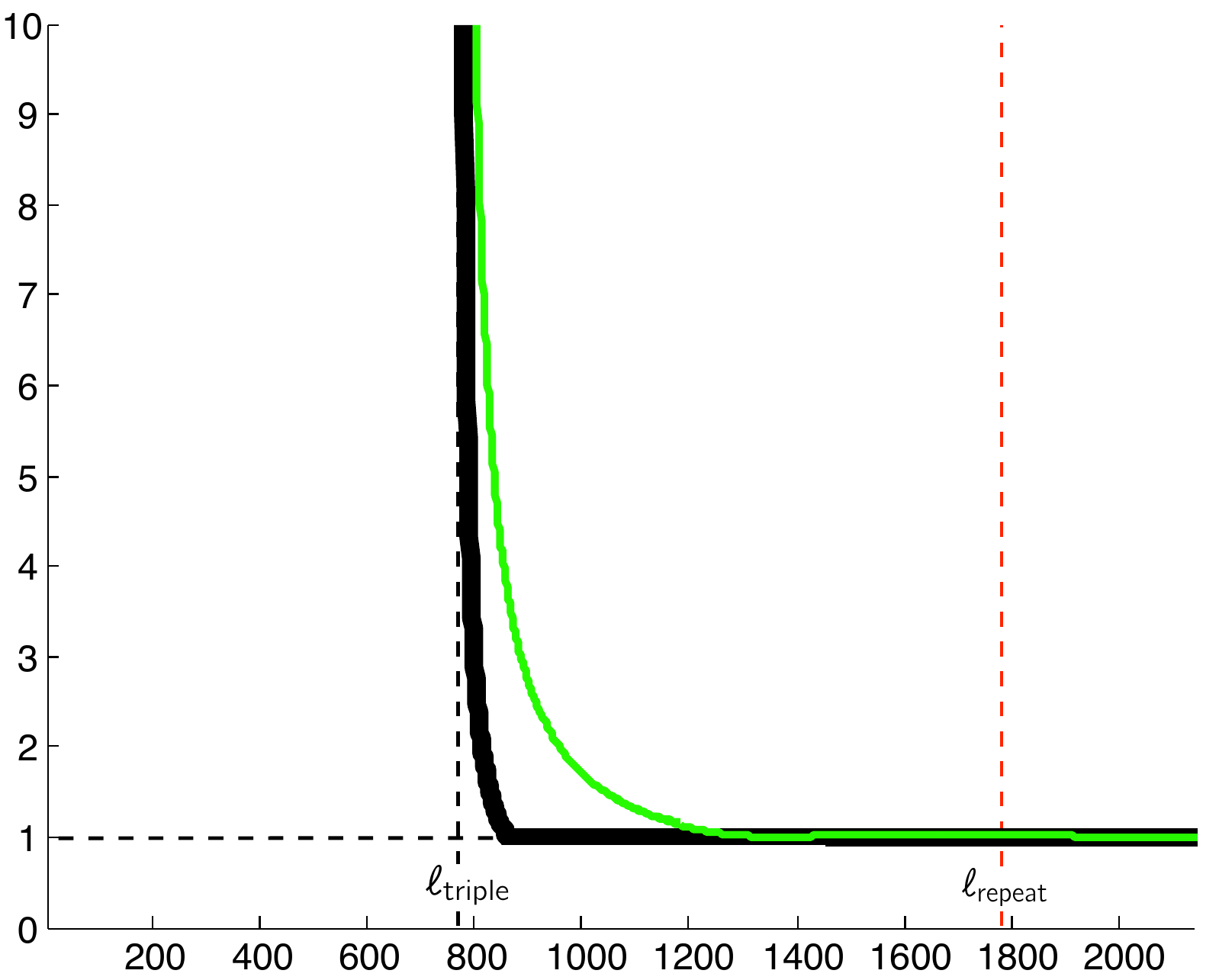}
\vspace{-2mm}
\caption{Performance of {\bdbII } on {\em{P Marinus}}, where $\Ltri > \Lint$.}
\label{fig:bridging2triplesINTRO}
\end{centering}
\vspace{-2mm}
\end{figure}

An interesting feature of the feasibility plots is that for typical repeat statistics exhibited by DNA data, the minimum coverage depth is characterized by a  {\em critical phenomenon}:
If the read length $L$ is below $\lcrit =\Lint$, reliable reconstruction of the DNA sequence is impossible no matter what the coverage depth is, but if the read length $L$ is slightly above $\lcrit$, then covering the sequence suffices, i.e. $\cb =c/\clw= 1$. 
The sharpness of the critical phenomenon is described by the size of the \emph{critical window}, which refers to the range of $L$ over which the transition from one regime to the other occurs. For the case when {\bdbII} is near optimal, the width $\wind$ of the window size can be well approximated as:
\vspace{-1mm}
\begin{equation}\label{e:simpleWindowINTRO}
\wind  \approx \frac{\lcrit}{2r+1}, \quad \mbox{where } r := \frac{\log \frac G\lcrit}{\log \eps\inv}\,.
\vspace{-1mm}
\end{equation}
For the  \hc{19} dataset, the critical window size evaluates to about $19 \%$ of $\lcrit$.

In Sections \ref{sec:Lower} and \ref{sec:opt_algo}, we discuss the underlying analysis and algorithm design supporting the plots.  The curves are all computed from formulas, which are validated by simulations in Section~\ref{sec:simulations}. We return in Section~\ref{sec:conclusions} to put our contributions in  a broader perspective and discuss extensions to the basic framework.  All proofs can be found in the appendix. 

\section{Lower bounds} 
\label{sec:asymptotic_analysis}
In this section we discuss lower bounds, due to coverage analysis and certain repeat patterns, on the required coverage depth and read length. The style of analysis here is continued in Section~\ref{sec:opt_algo}, in which we search for an assembly algorithm that performs close to the lower bounds. 

\label{sec:Lower}

\subsection{Coverage bound}

 Lander and Waterman's coverage analysis \cite{LW88} gives the well known condition for the number of reads $\nc$ required to cover the entire DNA sequence with probability at least  $1-\eps$. In the regime when $L \ll G$,  one may make the standard assumption that the starting locations of the $N$ reads follow a Poisson process with rate $\lam = N/G$, and  the number $\nc$
 is to a very good approximation given by the solution to the equation 
\begin{equation}
\label{e:cov}
\nc =  \frac GL \log \frac \nc\eps\,.
\end{equation}

The corresponding coverage depth is 
$\clw = \nc L/ G$.  This is our baseline coverage depth against which to compare the coverage depth of various algorithms. For each algorithm, we will plot 
$$\cb := \frac c{\clw} = \frac{N}{\nc}\,,$$
the coverage depth required by that algorithm normalized by $\cc$.  Note that $\cb$ is also the ratio of the number of reads $N$ required by an algorithm to $\nc$. The requirement $\cb \geq 1$ is due to the lower bound on the number of reads obtained by the Lander-Waterman coverage condition.

\subsection{Ukkonen's condition} \label{sec:Ukkonen}
A second constraint on reads arises from repeats. A lower bound on the read length $L$ follows from Ukkonen's condition \cite{Ukk92}:  
 if there are \emph{interleaved repeats} or \emph{triple repeats} in the sequence of length at least $L-1$, then the likelihood of observing the reads is the same for more than one possible DNA sequence and hence correct reconstruction is not possible. Fig.~\ref{fig:Ukkonen} shows an example with interleaved repeats. (Note that we assume $1-\eps>1/2$, so random guessing between equally likely sequences is not viable.) 
 
 \begin{figure}[htb]
\begin{centering}
\vspace{-4mm}
\includegraphics[width=2.3in]{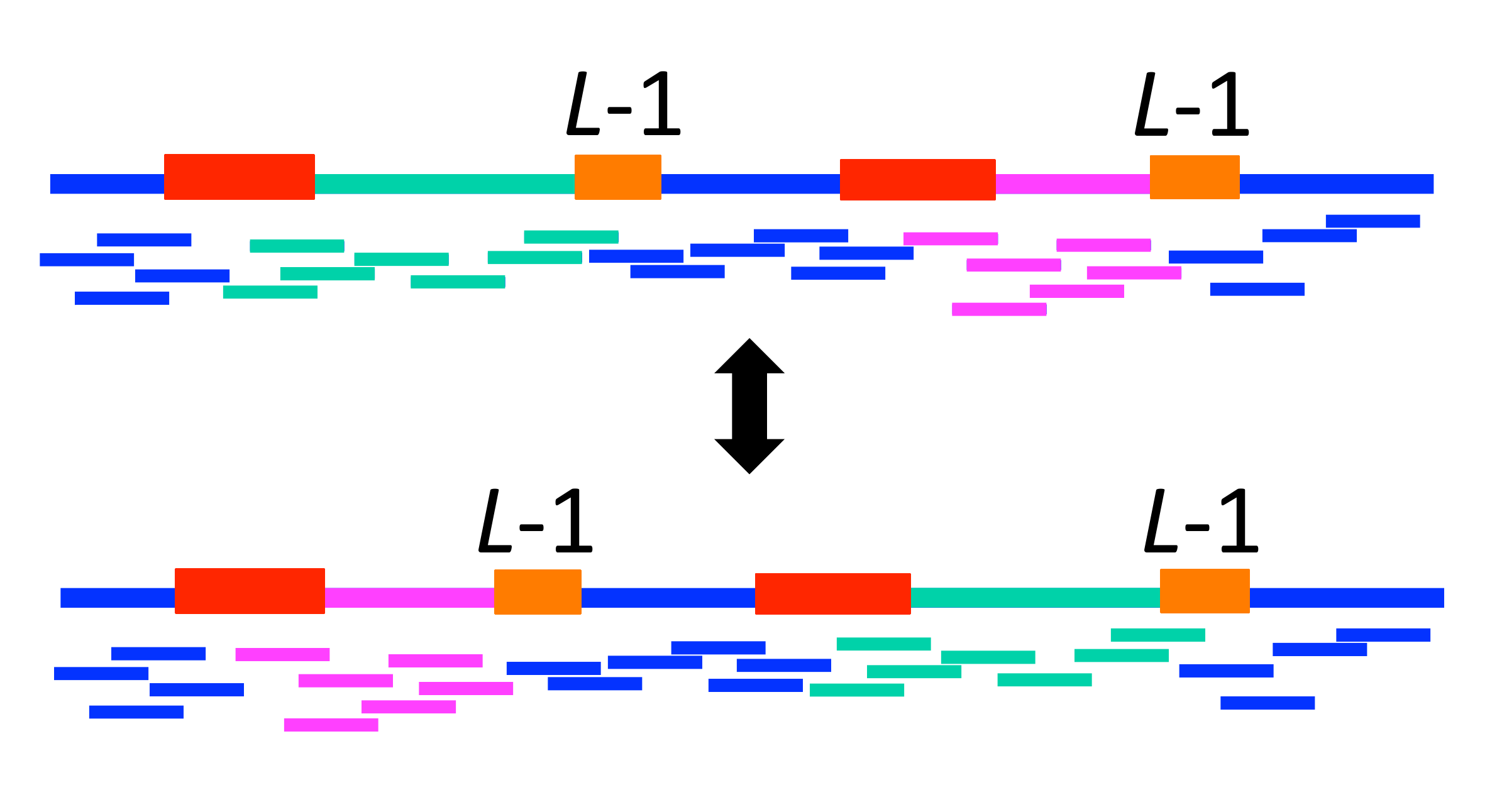}
\vspace{-5mm}
\caption{The likelihood of observing the reads under two possible sequences (the green and magenta segments swapped) is the same. Here, the two red subsequences form a repeat and the two orange subsequences form another repeat. 
} \label{fig:Ukkonen}
\end{centering}
\end{figure}

 We take a moment to carefully define the various types of repeats. Let $\s_t^\ell$ denote the length-$\ell$ subsequence of the DNA sequence ${\bf s}$ starting at position $t$. A \emph{repeat} of length $\ell$ is a subsequence appearing twice, at some positions $t_1,t_2$ (so $\s_{t_1}^{\ell} = \s_{t_2}^{\ell}$) that is maximal (i.e. $s(t_1-1)\neq s(t_2-1)$ and $s(t_1+\ell)\neq s(t_2+\ell)$). 
 Similarly, a \emph{triple repeat} of length $\ell$ is a subsequence appearing three times, at positions $t_1,t_2,t_3$, such that $\s_{t_1}^{\ell} = \s_{t_2}^{\ell}=\s_{t_3}^{\ell}$, and such that neither of $s(t_1-1)= s(t_2-1) = s(t_3-1)$ nor $s(t_1+\ell)= s(t_2+\ell)=s(t_3+\ell)$ holds\footnote{Note that a subsequence that is repeated $f$ times gives rise to $f\choose 2$ repeats and $f\choose 3$ triple repeats.}. 
A \emph{copy} is a single one of the instances of the subsequence's appearances. A \emph{pair} of repeats refers to two repeats, each having two copies. 
A pair of repeats, one at positions $t_1,t_3$ with $t_1<t_3$ and the second at positions $t_2,t_4$ with $t_2<t_4$, is \emph{interleaved} if $t_1<t_2<t_3<t_4$ or $t_2<t_1<t_4<t_3$ (Fig.~\ref{fig:Ukkonen}). The length of a pair of interleaved repeats is defined to be the length of the shorter of the two repeats.

Ukkonen's condition implies a lower bound on the read length, 
 $$
 \vspace{-1mm}
 L> \lcrit: = \max\{\Lint,\Ltri\} + 1\,.
 \vspace{-1mm}
 $$ 
 Here $\Lint$ is the length of the longest pair of interleaved repeats on the DNA sequence 
and $\Ltri$ is the length of the longest triple repeat.

\begin{figure}[htb]
\begin{centering}
\includegraphics[width=2.7in]{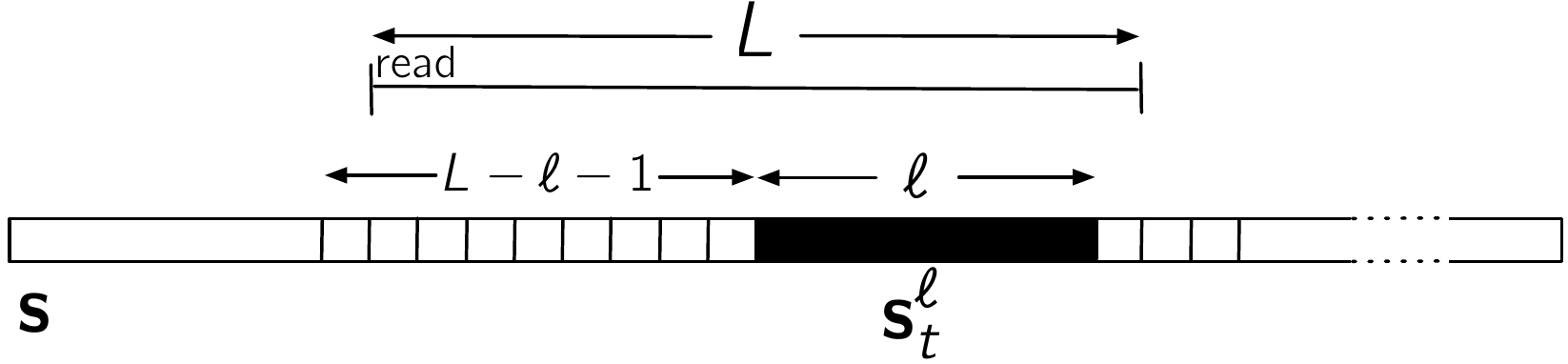}
\vspace{-2mm}
\caption{A subsequence $\s_t^\ell$ is bridged if and only if there exists at least one read which covers at least one base on both sides of the subsequence, i.e. the read arrives in the preceding length $L-\ell-1$ interval.} \label{fig:bridge}
\end{centering}
\vspace{-2mm}
\end{figure}

Ukkonen's condition says that for read lengths less than $\lcrit$, reconstruction is impossible no matter what the coverage depth is. But it can be generalized to provide a lower bound on the coverage depth for read lengths greater than $\lcrit$, through the important concept of {\em bridging} as shown in Figure~\ref{fig:bridge}.  We observe that in Ukkonen's interleaved or triple repeats, the actual length of the repeated subsequences is irrelevant; rather, to cause confusion it is enough that all the copies of the pertinent repeats are unbridged. This leads to the following theorem.

\begin{theorem}\label{t:Ukkonen}
Given a DNA sequence ${\bf s}$ and a set of reads, if there is a pair of  interleaved repeats or a triple repeat whose copies are all unbridged, then there is another sequence ${\bf s'}$ of the same length under which the likelihood of observing the reads is the same.
\end{theorem}

For brevity, we will call a repeat or a triple repeat {\em bridged} if at least one copy of the repeat is bridged,  and a pair of interleaved repeats {\em bridged } if at least one of the repeats is bridged. Thus, the above theorem says that a necessary condition for reconstruction is that all interleaved and triple repeats are bridged. 

How does Theorem \ref{t:Ukkonen} imply a lower bound on the coverage depth? Focus on the longest pair of interleaved repeats and suppose the read length $L$ is between the lengths of the shorter and the longer repeats. The probability this pair is unbridged is $(\pub_{\Lint})^2$, where
\vspace{-1mm}
\begin{align}
\nonumber
\pub_\ell &:=  \P[\mbox{$\ell$-length subseq. is unbridged}]\\
\label{e:pub}
& =  e^{\frac{N}{G}(L-\ell - 1)^+}.
\end{align}
Theorem \ref{t:Ukkonen} implies that the probability of making an error in the reconstruction is at least $1/2$ if this event occurs. Hence, the requirement that $\Pe \le \epsilon$ implies a lower bound on the number of reads $N$:
\begin{equation}
N \ge \frac{G}{(L-\Lint-1)\ln(1/(2\epsilon))}.
\end{equation}
A similar lower bound can be derived using the longest triple repeat. A slightly tighter lower bound can be obtained by taking into consideration the bridging of  {\em all} the interleaved and triple repeats, not only the longest one, resulting in the black  curve in Fig.~\ref{fig:plotINTRO}.

\vspace{-2mm}
\section{Towards optimal assembly }
\label{sec:opt_algo}
We now begin our search for algorithms performing close to the lower bounds derived in the previous section. 
Algorithm assessment begins with obtaining deterministic sufficient conditions for success in terms of repeat-bridging. We then find the necessary $N$ and $L$ in order to satisfy these sufficient conditions with a target probability $1-\eps$. The required coverage depth for each algorithm depends only on certain repeat statistics extracted from the DNA data, which may be thought of as  \emph{sufficient statistics}. 

\subsection{Greedy algorithm}\label{sec:Greedy}


The greedy algorithm, denoted {\GR}, with pseudocode in section \ref{sec:greedyApp}, is described as follows. Starting with the initial set of reads, the two fragments (i.e. subsequences) with maximum length overlap are merged, and this operation is repeated until a single fragment remains. Here the overlap of two fragments $\bx,\by$ is a suffix of $\bx$ equal to a prefix of $\by$, and merging two fragments results in a single longer fragment. 


\begin{theorem}
 \GR { }reconstructs the original sequence ${\bf s}$ if every repeat is bridged. 
\label{thm:GreedyTheorem}
\end{theorem}




Theorem~\ref{thm:GreedyTheorem}  allows us to determine the coverage depth required by {\GR }: we must ensure that all repeats are bridged. By the union bound,
\begin{equation}
\label{e:ub}
\P[\mbox{some repeat is unbridged}] \le \sum_m a_m \left (\pub_m \right)^2\,,
\end{equation} 
where $\pub_m$ is defined in (\ref{e:pub}) and $a_m$ is the number of repeats of length $m$. 
Setting the right-hand side of (\ref{e:ub}) to $\eps$ ensures $\Pe \le \eps$ and yields the performance curve of {\GR} in Fig. \ref{fig:plotINTRO}. Note that the repeat statistics $\{a_m\}$ are sufficient to compute this curve.


{\GR } requires $L> \Lrep + 1$, whereas the lower bound has its asymptote at $L=\Lint + 1$. In chromosome 19, for instance, there is a large difference between $\Lint = 2248$ and $\Lrep = 4092$, and in Fig~\ref{fig:plotINTRO} we see a correspondingly large gap. $\GR$ is evidently sub-optimal in handling interleaved repeats. 
Its strength, however, is that once the reads are slightly longer than $\Lrep$, coverage of the sequence is sufficient for correct reconstruction. 
Thus if $\Lrep\approx \Lint$, then {\GR } is close to optimal.

\subsection{$K$-mer algorithms}

The greedy algorithm fails when there are unbridged repeats, even if there are no unbridged \emph{interleaved} repeats, and therefore requires a read length much longer than that required by Ukkonen's condition. As we will see, $K$-mer algorithms do not have this limitation.



\subsubsection{Background}\label{sec:KmerBackgound}
In the introduction we mention Sequencing By Hybridization (SBH), for which Ukkonen's condition was originally introduced. In the SBH setting, an optimal algorithm matching Ukkonen's condition is known, due to Pevzner \cite{Pev95}. 



  Pevzner's algorithm is based on finding an appropriate cycle in a $K$-mer graph (also known as a de Bruijn graph) with $K=L-1$ (see e.g. \cite{PPT11} for an overview). A $K$-mer graph is formed 
   by first creating a node in the graph for each unique $K$-mer (length $K$ subsequence) in the set of reads, and then adding an edge with overlap $K-1$ between any two nodes representing $K$-mers that are {\emph{adjacent}} in a read, i.e. offset by a single nucleotide. Edges thus correspond to unique $(K+1)$-mers in $\s$ and paths correspond to longer subsequences obtained by merging the constituent nodes.  
There exists a cycle corresponding to the original sequence $\s$, and reconstruction entails finding this cycle. 
 
As is common, we will replace edges corresponding to an unambiguous path by a single node (c.f. Fig.~\ref{fig:condensing}). Since the subsequences at some nodes are now longer than $K$, this is no longer a $K$-mer graph, and we call the more general graph a sequence graph. 
The simplified graph is called the \emph{condensed sequence graph}. 

\begin{figure}[bth]
\begin{centering}
\includegraphics[width=2.5in]{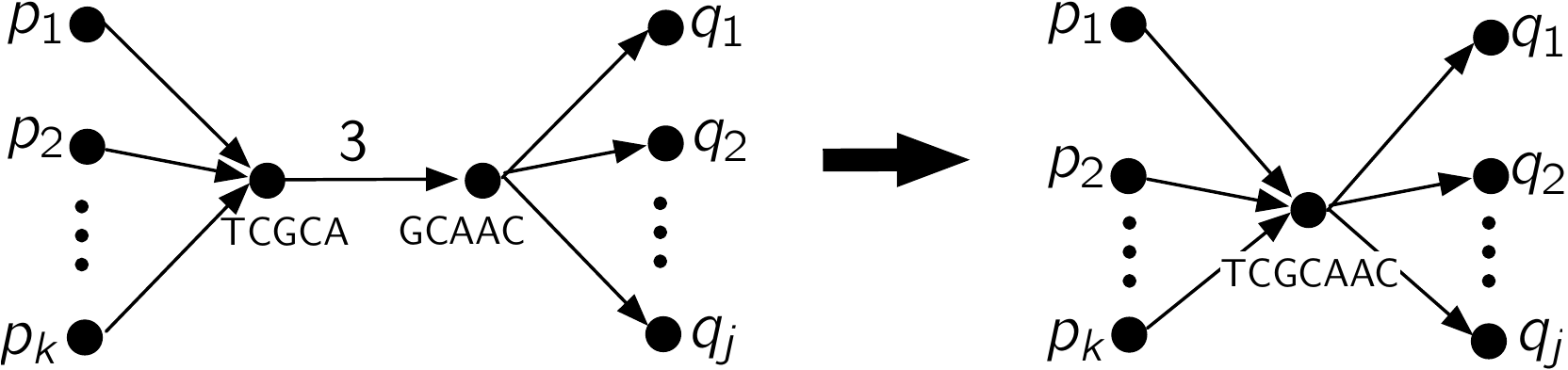}
\vspace{-1mm}
\caption{Contracting an edge by merging the incident nodes. Repeating this operation results in the condensed graph.} \label{fig:condensing}
\end{centering}
\vspace{-2mm}
\end{figure}

The condensed graph has the useful property that if the original sequence $\s$ is reconstructible, then $\s$ is determined by a unique Eulerian cycle:  
\begin{theorem}
\label{t:SBH_no_multiplicities}
Let $\G_0$ be the $K$-mer graph constructed from the $(K+1)$-spectrum $\calS_{K+1}$ of $\s$, and 
let $\G$ be the condensed sequence graph obtained from $\G_0$. If Ukkonen's condition is satisfied, i.e. there are no triple or interleaved repeats of length at least $K$, then there is a unique Eulerian cycle $\CC$ in $\G$ and $\CC$ corresponds to $\s$.
\end{theorem}

Theorem~\ref{t:SBH_no_multiplicities} characterizes, deterministically, the values of $K$ for which reconstruction from the $(K+1)$-spectrum is possible. 
We proceed with application of the $K$-mer graph approach to shotgun sequencing data.

\subsubsection{Basic $K$-mer algorithm}

Starting with Idury and Waterman \cite{IW95}, and then Pevzner et al.'s \cite{PTW01} \textsc{euler} algorithm, most current assembly algorithms for shotgun sequencing are based on the $K$-mer graph.
Idury and Waterman \cite{IW95} made the key observation that SBH with subsequences of length $K+1$ can be \emph{emulated} by shotgun sequencing if each read overlaps the subsequent read by $K$: the set of all $(K+1)$-mers within the reads is equal to 
the $(K+1)$-spectrum $\calS_{K+1}$. 
The resultant algorithm \textsc{DeBruijn} which consists of constructing the $K$-mer graph from the $(K+1)$-spectrum observed in the reads, condensing the graph, and then identifying an Eulerian cycle, has 
sufficient conditions for correct reconstruction as follows. 

%

%

\begin{theorem}
\label{l:lemmaAbody}
 \textsc{DeBruijn} with parameter choice $K$ reconstructs  the original sequence $\s$ if:
\begin{enumerate}
\item[(a)] 
$K> \Lint$
\item[(b)] 
$ K> \Ltri$ 
\item[(c)] adjacent reads overlap by at least K
\end{enumerate}
\end{theorem}


Lander and Waterman's coverage analysis  applies also to Condition (c) of Theorem~\ref{l:lemmaAbody}, yielding a normalized coverage depth requirement $\cb = 1/(1-K/L)$.  The larger the overlap $K$, the higher the coverage depth required. Conditions (a) and (b) say that the smallest $K$ one can choose is $K=\max\{\Ltri,\Lint\}+1$, so
\begin{equation}
\label{e:dbPerf}
\cb =  \frac{1}{1-\frac {\max\{\Ltri,\Lint\}+1}{L}}\,.
\end{equation}
%

The performance of \textsc{DeBruijn} is plotted in Fig.~\ref{fig:plotINTRO}. \textsc{DeBruijn} significantly improves on \textsc{Greedy} by obtaining the correct first order performance: given sufficiently many reads, the read length $L$ may be decreased to $\max\{\Ltri,\Lint\}+1$. Still, the number of reads required to approach this critical length is far above the lower bound. The following subsection pursues reducing $K$ in order to reduce the required number of reads. 


\vspace{-2mm}
\subsection{Improved $K$-mer algorithms}\label{sec:ImprovedKmers}
Algorithm \textsc{DeBruijn} ignores a lot of information contained in the reads, and indeed all of the $K$-mer based algorithms proposed by the sequencing community (including \cite{IW95}, \cite{PTW01}, \cite{Sim09}, \cite{Gne11}, \cite{Maccallum:2009qy}, \cite{Zerbino:2008fj}) use the read information to a greater extent than the naive \textsc{DeBruijn} algorithm. Better use of the read information, as described below in algorithms {\bdbI }  and {\bdbII}, will allow us to relax the condition $K> \max\{\Lint,\Ltri\}$ for success of \textsc{DeBruijn}, which in turn reduces the high coverage depth required by Condition (c).

Existing algorithms use read information in a variety of distinct ways to resolve repeats. For instance, Pevzner et al. \cite{PTW01} observe that for graphs where each edge has multiplicity one, if one copy of a repeat is bridged, the repeat can be resolved through what they call a ``detachment''. The algorithm {\bdbI } described below is very similar, and resolves repeats with two copies if at least one copy is bridged. 

Meanwhile, other algorithms are better suited to higher edge multiplicities due to higher order repeats; IDBA (Iterative DeBruijn Assembler) \cite{peng2010idba} creates a series of $K$-mer graphs, each with larger $K$, and at each step uses not just the reads to identify adjacent $K$-mers, but also all the unbridged paths in the $K$-mer graph with smaller $K$. Although not stated explicitly in their paper, we observe here that if all copies of every repeat are bridged, then IDBA correctly reconstructs. 

However, it is suboptimal to require that \emph{all} copies of every repeat up to the maximal $K$ be bridged. We  introduce {\bdbII }, which combines the aforementioned ideas to simultaneously allow for single-bridged double repeats, triple repeats in which all copies are bridged, and unbridged non-interleaved repeats. 

\vspace{-2mm}
\subsubsection{{\sc \bdbI}}

{\bdbI } improves on {\db } by resolving bridged 2-repeats (i.e. a repeat with exactly two copies in which at least one copy is bridged by a read). 
Condition (a) $K>\Lint$ for success of {\db } (ensuring that no interleaved repeats appear in the initial $K$-mer graph) is updated to require only no \emph{unbridged} interleaved repeats, which matches the lower bound. With this change, Condition (b) $K>\Ltri$ forms the bottleneck for typical DNA sequences. Thus {\bdbI }  is optimal with respect to interleaved repeats, but it is suboptimal with respect to triple repeats.

{\bdbI } deals with repeats by performing surgery on certain nodes in the sequence graph. In the sequence graph, a repeat corresponds to a node we call an \emph{\xn}, a node with in-degree and out-degree each at least two (e.g. Fig.~\ref{fig:bridgingStepEx}). A self-loop adds one each to the in-degree and out-degree. The cycle $\CC(\s)$ traverses each {\xn } at least twice, so {\xns } correspond to repeats in $\s$. We call an {\xn } traversed exactly twice a {2-\xn }; 
these nodes correspond to 2-repeats, and are said to be bridged if the corresponding repeat in $\s$ is bridged.

%

In the repeat resolution step of {\bdbI }  (illustrated in Fig.~\ref{fig:bridgingStepEx}), bridged {2-\xns }are duplicated in the graph and incoming and outgoing edges are inferred using the bridging read, reducing possible ambiguity. 

\vspace{-2mm}
\begin{figure}[htb]
\begin{centering}
\includegraphics[width=3in]{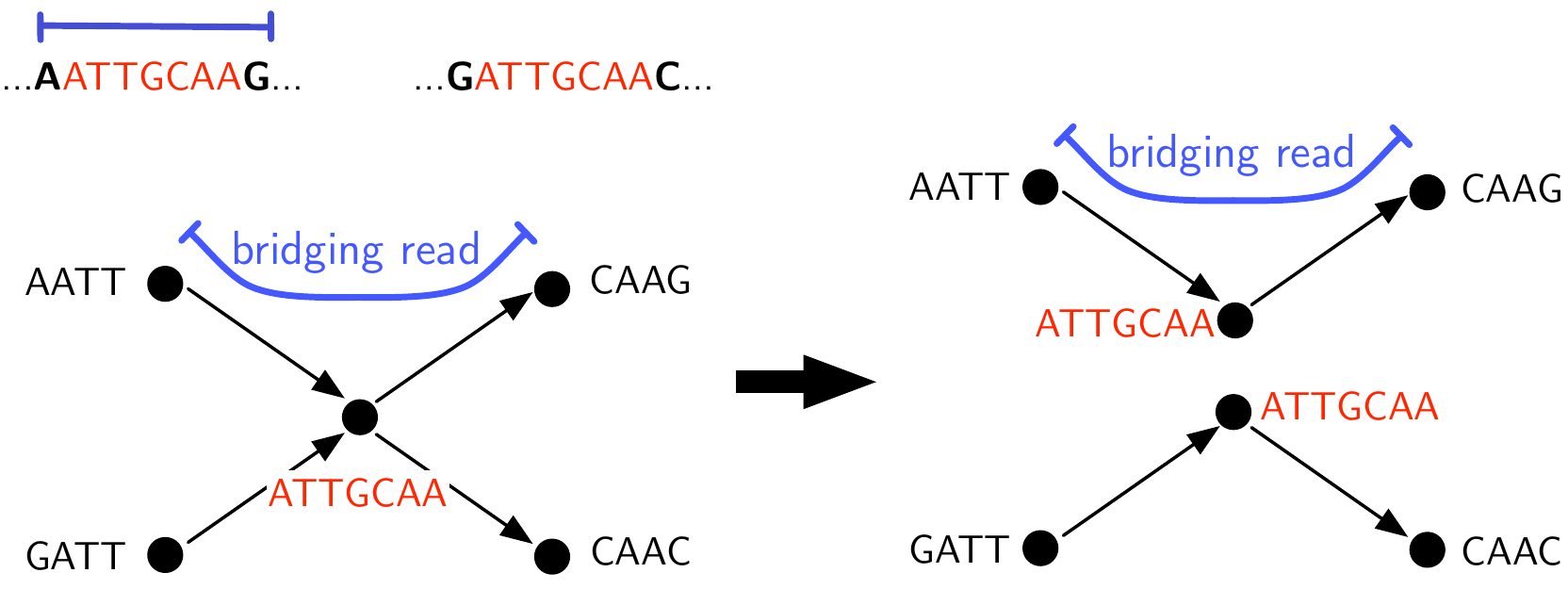}
\caption{An example of the bridging step in {\bdbI}.} \label{fig:bridgingStepEx}
\end{centering}
\vspace{-2mm}
\end{figure}



\vspace{-3mm}
\begin{theorem}
\label{l:bridgingDeBruijn}
{\bdbI }  with parameter choice $K$ reconstructs the original sequence $\s$ if:
\begin{enumerate}
\item[(a)] all interleaved repeats are bridged
\item[(b)] 
$K>\Ltri$
\item[(c)] adjacent reads overlap by at least K.
\end{enumerate}
\end{theorem}

By the union bound,
\begin{align}
\nonumber
&  \P[\mbox{some interleaved repeat is unbridged}] \\
\label{e:int_ub}
&\le  \sum_{m,n} b_{m,n} \left(\pub_m\right)^2 \left(\pub_n\right)^2
\end{align}
where $b_{m,n}$ is the number of interleaved repeats in which one repeat is of length $m$ and the other is of length $n$. To ensure that condition (a) in the above theorem fails with probability no more than $\eps$, the right hand side of (\ref{e:int_ub}) is set to be  $\eps$; this imposes a constraint on the coverage depth. Furthermore, conditions (b) and (c) imply that the normalized coverage depth 
 $\cb \ge  1/(1-(\Ltri+1)/L)$. These two constraints together yield the performance curve of  {\bdbI } in
Figure~\ref{fig:plotINTRO}.

\vspace{-2mm}
\subsubsection{{\sc \bdbII}}
\label{sec:bdbII}

We now turn to triple repeats. As previously observed, it can be challenging to resolve repeats with more than one copy  \cite{PTW01}, because an edge into the repeat may be paired with more than one outgoing edge. As discussed above, our approach here shares elements with IDBA   \cite{peng2010idba}: we note that increasing the node length serves to resolve repeats. Unlike IDBA, we do not increase the node length globally.

As noted in the previous subsection, repeats correspond to nodes in the sequence graph we call \emph{{\xns }}. Here the converse is false: not all repeats correspond to {\xns}. 
A repeat is said to be \emph{all-bridged} if \emph{all} repeat copies are bridged, and an {\xn } is called all-bridged if the corresponding repeat is all-bridged. 

\begin{figure}[htb]
\vspace{-2mm}
\centering{
\includegraphics[width=2.7in]{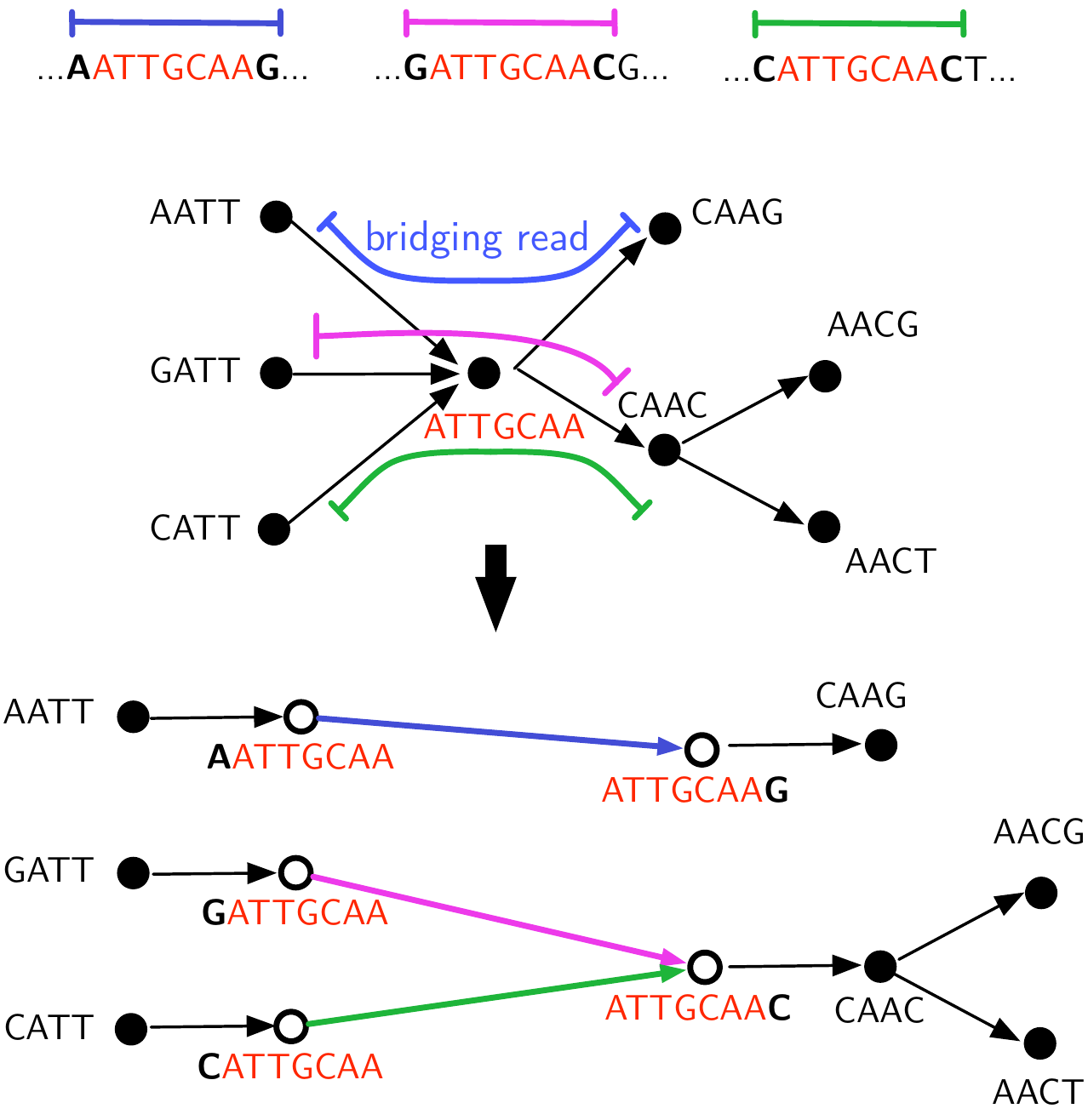}
 }
 \vspace{-3mm}
\caption{{\bdbII } resolves an {\xn } with label ATTGCAA corresponding to a triple repeat.}
\label{fig:XnodeA}
\vspace{-4mm}
\end{figure}

The requirement that triple repeats be all-bridged allows them to be resolved \emph{locally} (Fig.~\ref{fig:XnodeA}). The {\xn } resolution procedure given in Step~4 of {\bdbII } 
can be interpreted in the $K$-mer graph framework as increasing $K$ locally so that repeats do not appear in the graph. In order to do this, we introduce the following notation for extending nodes: 
Given an edge  $(\bv,\bq)$ with weight $a_{\bv,\bq}$, let $\vq$ denote $\bv$ extended one base to the right along $(\bv,\bq)$, 
 i.e. $\vq = \bv \, \bq_{\olap\bv\bq+1}^{1}$ (notation introduced in Sec.~\ref{sec:Ukkonen}).  Similarly, let $\pv = \bp_{\vecend-\olap\bp\bv}^{1} \,  \bv $. {\bdbII } is described as follows.

\begin{algorithm}[H]
\caption{{\sc \bdbII}. Input: reads $\RR$, parameter $K$. Output: sequence $\sh$. }
\label{alg3}
\emph{$K$-mer steps 1-3:}\\
1. For each subsequence $\x$ of length $K$ in a read, form a node with label $\x$.  \\
2. For each read, add edges between nodes representing adjacent $K$-mers in the read. \\
3. Condense the graph (c.f. Fig.~\ref{fig:condensing}). \\ 
4.~\emph{Bridging step:} (See Fig.~\ref{fig:XnodeA}). While there exists a bridged {\xn } $\bv$: (i) For each edge $(\bp_i,\bv)$ with weight $a_{\bp_i,\bv}$, create a new node $\bu_i = \pvinp{\bp_i}\bv$ 
and an edge $(\bp_i, \bu_i)$ with weight $1+a_{\bp_i,\bv}$. Similarly for each edge $(\bv, \bq_j)$, create a new node $\bw_j=\vqinp\bv{\bq_j}$ and edge $(\bw_j,\bq_j)$. 
(ii) If $\bv$ has a self-loop $(\bv,\bv)$ with weight $a_{\bv,\bv}$, add an edge $(\vqinp\bv\bv,\pvinp\bv\bv)$ with weight $a_{\bv,\bv}+2$. 
(iii) Remove node $\bv$ and all incident edges.
(iv) For each pair $\bu_i,\bw_j$ adjacent in a  read, 
add edge $(\bu_i,\bw_j)$. If exactly one each of the $\bu_i$ and $\bw_j$ nodes have no added edge, add the edge.   (v) Condense graph. 
\\
5.~\emph{Finishing step:} Find an Eulerian cycle in the graph and return the corresponding sequence.
\end{algorithm}

\begin{theorem}
\label{l:dogBones}
 The algorithm {\bdbII }  reconstructs the sequence $\s$ if:
\begin{enumerate}
\item[(a)] all interleaved repeats are bridged
\item[(b)] all triple repeats are \textbf{all-bridged}
\item[(c)] the sequence is covered by the reads.
\end{enumerate}
\end{theorem}
A similar analysis as for {\bdbI} yields the performance curve of {\bdbII} in Figure~\ref{fig:plotINTRO}.

\comments{
\subsection{Hybrid algorithm}
\texttt{\color{red} FIXME: discussion needs to be modified in light of good performance of K-mer algo}\\ \\
Examining Figs.~\ref{fig:bridging2Ch19} and \ref{fig:greedyCh19} reveals that given sufficient coverage depth, {\bdbII }  achieves the read length lower bound, $L>\max\{\Lint,\Ltri\}$, while if $L>\Lrep$ then \GR achieves the lower bound $\cb\geq1$ on the coverage depth. It would be desirable, however, to \emph{simultaneously} achieve both of the lower bounds, since this necessitates both the smallest read length $L$ and the least coverage depth. Motivated by the optimality of each of the two algorithms in the respective regimes, we seek to combine the two algorithms into a \emph{hybrid} algorithm  achieving the best of both worlds.

\subsubsection{Kmer-graph based Hybrid algorithm (I or II)}
Starting with the de Bruijn graph-based algorithm {\bdbI }  or \textsc{II}, we wish to relax condition (c), that each read overlaps its successor by at least $K$. If condition (c) is not satisfied, then some of the $K$-mers at the \emph{ends} of reads will not have edges in the de Bruijn graph. The natural next step is to \emph{greedily} add edges between $K$-mers lacking an edge. Of course, adding edges between $K$-mers at this point is the same as adding edges between reads since the reads overlap by less than $K$. Now, we would still like to resolve dogbones using the bridging read information, and this is done in Step~5 of \textsc{Hybrid}. But because reads do not always overlap one another by at least $K$, the existence of a dogbone in the graph corresponding to each repeat is not guaranteed. In order that dogbones appear for each repeat we add the condition that every maximal repeat of length $\geq K$ has its first and last $(K-1)$-mer bridged in all copies. \\

\begin{mdframed}
\textbf{Algorithm}{\textsc{Hybrid I/II}}. Input: reads $\RR$. Output: sequence $\sh$. \\ 
\noindent{\emph{$K$-mer steps:}\\
1. For each substring $\x$ of length $K$ in a read, form a node with label $\x$.  \\
2. For each read, add edges between nodes representing adjacent $K$-mers in the read. \\
\emph{Greedy steps:}\\
3. Consider all pairs of nodes $\x_1,\x_2$ in $\G$ satisfying $\dout(\x_1)=\din(\x_2)=0$, and add an edge $(\x_1,\x_2)$ with largest value $\ov(\x_1,\x_2)$. \\
4. Repeat Step 2 until no candidate pair of nodes remains.\\
5.~\emph{Bridging step:}~Step 3 from algorithm \textsc{BridgingDeBruijn I or II}.\\
6.~\emph{Finishing step:} Find a SE-cycle in $\G$. Return the sequence corresponding to the cycle.}\\
\textbf{Sufficient conditions for reconstruction:} (a) no \emph{unbridged} interleaved  maximal repeats,
(b) same as condition (b) of \textsc{BridgingDeBruijn I or II},
(c) modified condition on overlaps:
(i) no unbridged repeats of length $\leq K$ and (ii)
every maximal repeat of length $\geq K$ has its last and first $(K-1)$-mer bridged in all copies.  \\
\textbf{Sufficient conditions on genome statistics:} (a) $\cb \geq \max_m \frac{\be_m}{2(2-\frac mL)}$ and $L>\Lint$, (b) same as condition (b) of \textsc{BridgingDeBruijn I or II} (c) i. $\cb \geq \max_{m\leq K} \frac{\al_m}{2(1-\frac mL)}$, ii.
 $\cb>\frac{\log(4 \sum_{\ell \geq K} a_\ell )}{\log N} \frac{1}{1-\frac KL}$.
\end{mdframed}

\begin{figure}[htb]
\begin{centering}
\caption{Performance of \textsc{Hybrid II} on Chrom19.} \label{fig:hybrid2Ch19}
\end{centering}
\end{figure}

\begin{lemma}
\label{l:hybrid2}
Suppose the sequence $\s$ and reads satisfy the following assumptions:
\begin{enumerate}
\item[(a)] interleaved repeats: no \textbf{unbridged} interleaved maximal repeats
\item[(b)] triple repeats: no triple repeats of length $\geq$ K which are not \textbf{all-bridged}
\item[(c)] coverage depth: (i) no unbridged repeats of length $\leq K$ and (ii)
every maximal repeat of length $\geq K$ has its last and first $(K-1)$-mer bridged in all copies.
\end{enumerate}
Then \textsc{Hybrid II} with parameter choice $K$ successfully reconstructs the sequence $\s$.
\end{lemma}
\begin{proof}
The proof is contained in the appendix.
\end{proof}

\subsubsection{Read-overlap hybrid algorithm I}

\begin{mdframed}

\textbf{Algorithm}{\textsc{K-edge read-overlap based}}. Input: reads $\RR$. Output: sequence $\sh$. \\ 
\noindent{\emph{$K$-overlap steps:}\\
1. For each read, form a node with label $\x$.  \\
2. For each read, add edges to reads overlapping by at least $K$ bases. \\
3. Consider an edge to be \emph{necessary} if it is the only edge entering a read, or the only edge exiting a read (i.e. $(\x_1,\x_2)$ is necessary if  $\dout(\x_1)=1 $ or $\din(\x_2)=1$). Then for any read with multiple outgoing edges, keep only the necessary outgoing edges, and likewise for a read with multiple incoming edges, keep only necessary incoming edges.\\
4. Any time two reads $\x_1$ and $\x_2$ each have two outgoing edges to the same pair of reads $\x_3$,$\x_4$, create a new node $\x_5$ with label formed by the union of $\ov(\x_1,\x_2), \ov(\x_3, \x_4)$, and replace edges with $(\x_1, \x_5),(\x_2, \x_5), (\x_5, \x_3), (\x_1, \x_4)$.\\
5.~\emph{Finishing step:} Find a SE-cycle in $\G$. Return the sequence corresponding to the cycle.}\\
\textbf{Sufficient conditions for reconstruction:} (a) no \emph{unbridged} interleaved  maximal repeats,
(b) no triple repeats of length at least $K$,
(c) condition on overlaps at repeat boundary: every read overlaps its successor by at least $K$, if a read has multiple successors (and likewise for predecessors).  \\
\textbf{Sufficient conditions on genome statistics:} (a) $\cb \geq \max_m \frac{\be_m}{2(2-\frac mL)}$ and $L>\Lint$, (b) $L>K>\Ltri$ (c) \texttt{fill this in-- is it the same as before?}.
\end{mdframed}

\subsubsection{Read-overlap graph based hybrid algorithm II}
}

\vspace{-2mm}
\subsection{Gap to lower bound} 
\label{sec:gap}

The only difference between the sufficient condition guaranteeing the success of {\bdbII} and the necessary condition of the lower bound is the bridging condition of {\em triple} repeats: while {\bdbII} requires bridging {\em all three copies} of the triple repeats, the necessary condition requires only bridging {\em a single copy.} When $\Ltri$ is significantly smaller than   $\Lint$, the bridging requirement of interleaved repeats dominates over that of triple repeats and {\bdbII} achieves very close to the lower bound. This occurs in \hc{19} and the majority of the datasets we looked at. (See Fig.~\ref{fig:sim} and the plots in the supplementary material.) A critical phenomenon occurs as $L$ increases: for $L<\lcrit$ reconstruction is impossible, over a small critical window the bridging requirement of interleaved repeats (primarily the longest) dominates, and then for larger $L$, 
coverage suffices.

On the other hand, when $\Ltri$ is comparable or larger than $\Lint$, then {\bdbII} has a gap in the  coverage depth to the lower bound (see for example Fig.~\ref{fig:bridging2triplesINTRO}).  If we further assume that the longest triple repeat is dominant, then this gap can be calculated to be a factor of  
$  3\cdot \frac{\log3\eps\inv}{\log \eps\inv}\approx 3.72 \quad \text{for } \eps = 10^{-2}.$
This gap occurs only within the critical window where the repeat-bridging constraint is active. Beyond the critical window, the coverage constraint dominates and {\bdbII} is optimal. Further details are provided in the appendices.

\vspace{-2mm}
\section{Simulations and complexity}
\label{sec:simulations}

\newcommand{\simfigsize}{1.8in}

\begin{figure*}[t!]
\subfloat[S. Aureus ]{
\includegraphics[height=\simfigsize]{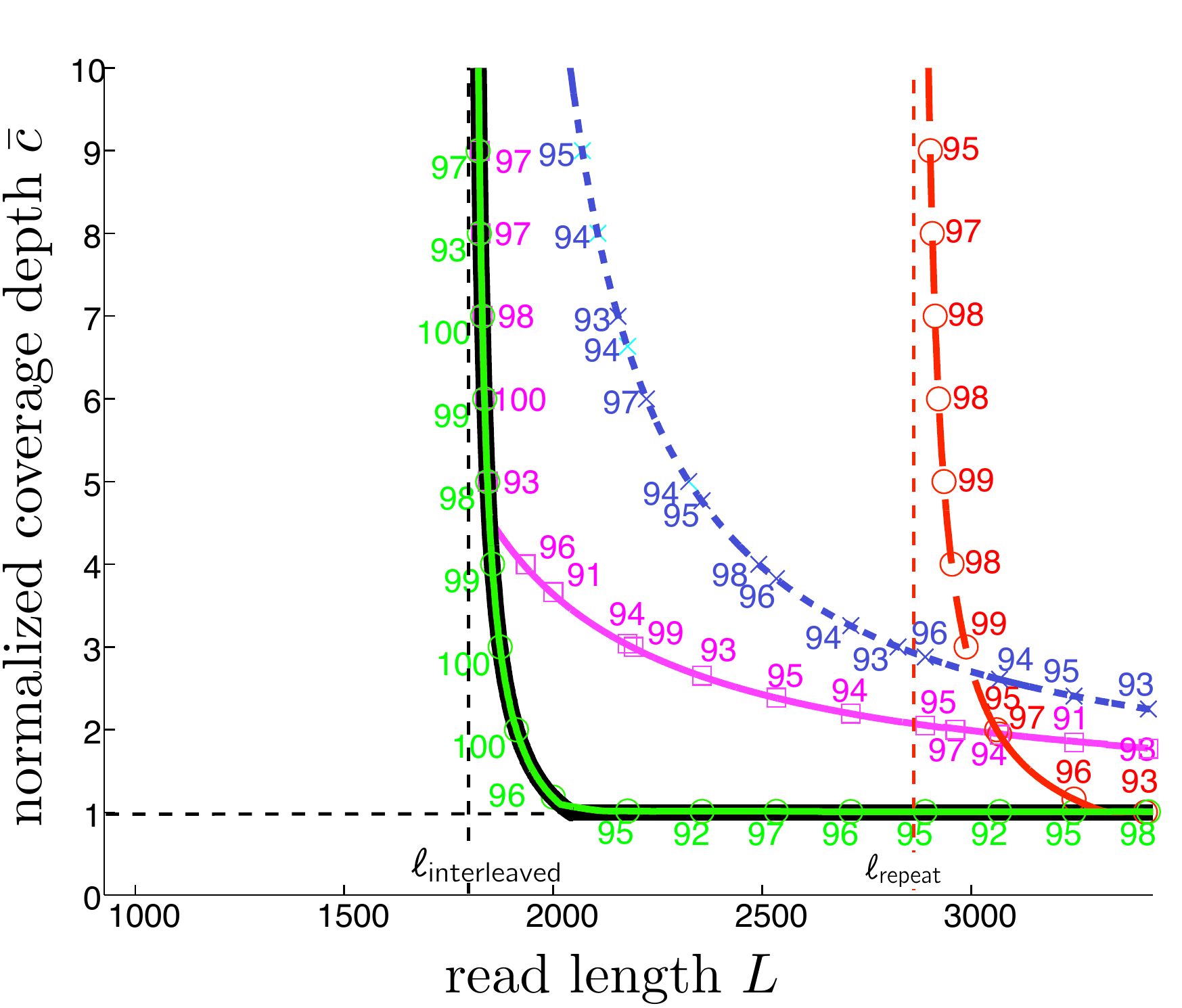}}
\subfloat[R. sphaeroides]{
\includegraphics[height=\simfigsize]{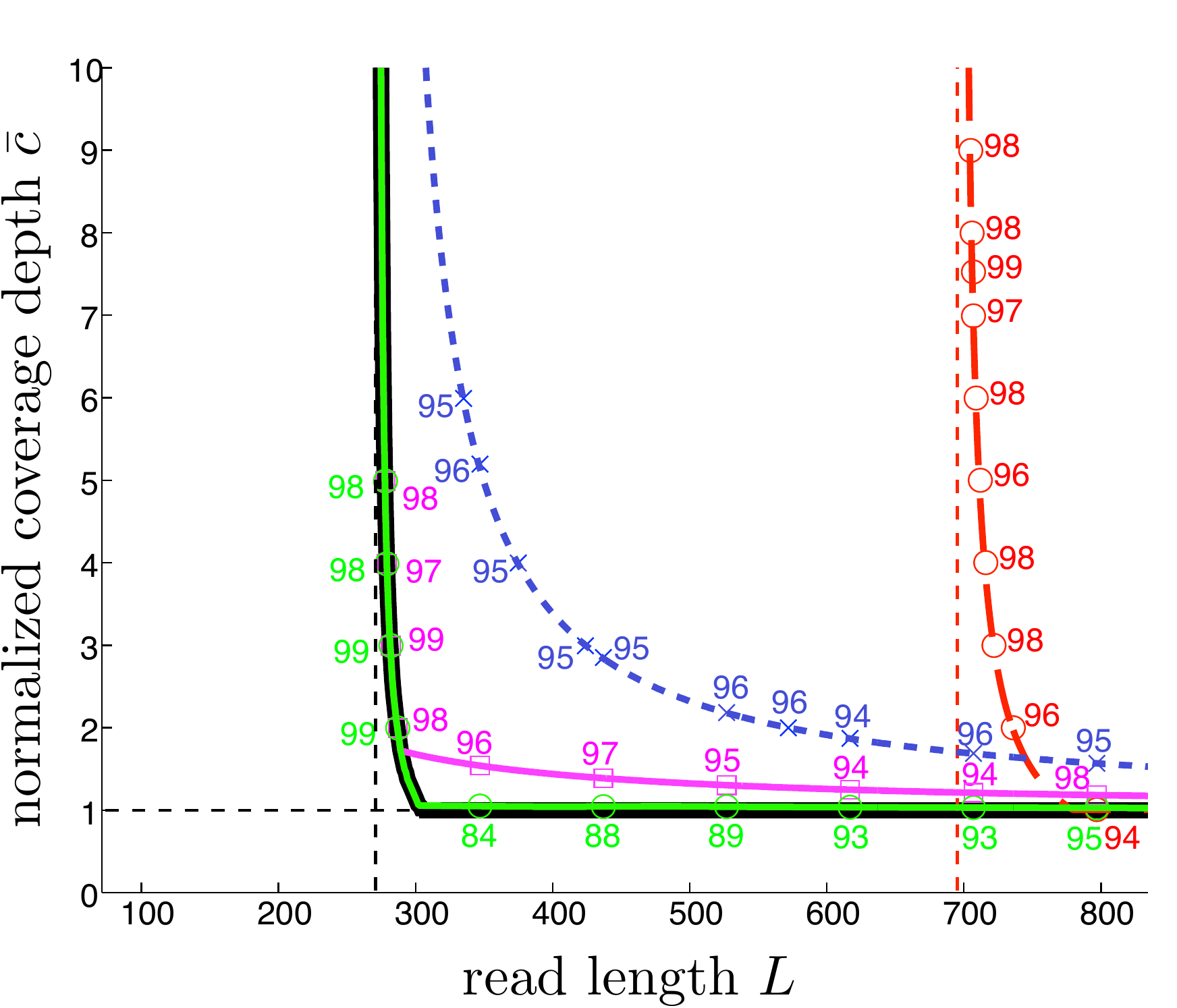} }
\subfloat[\hc{14}]{
\includegraphics[height=\simfigsize]{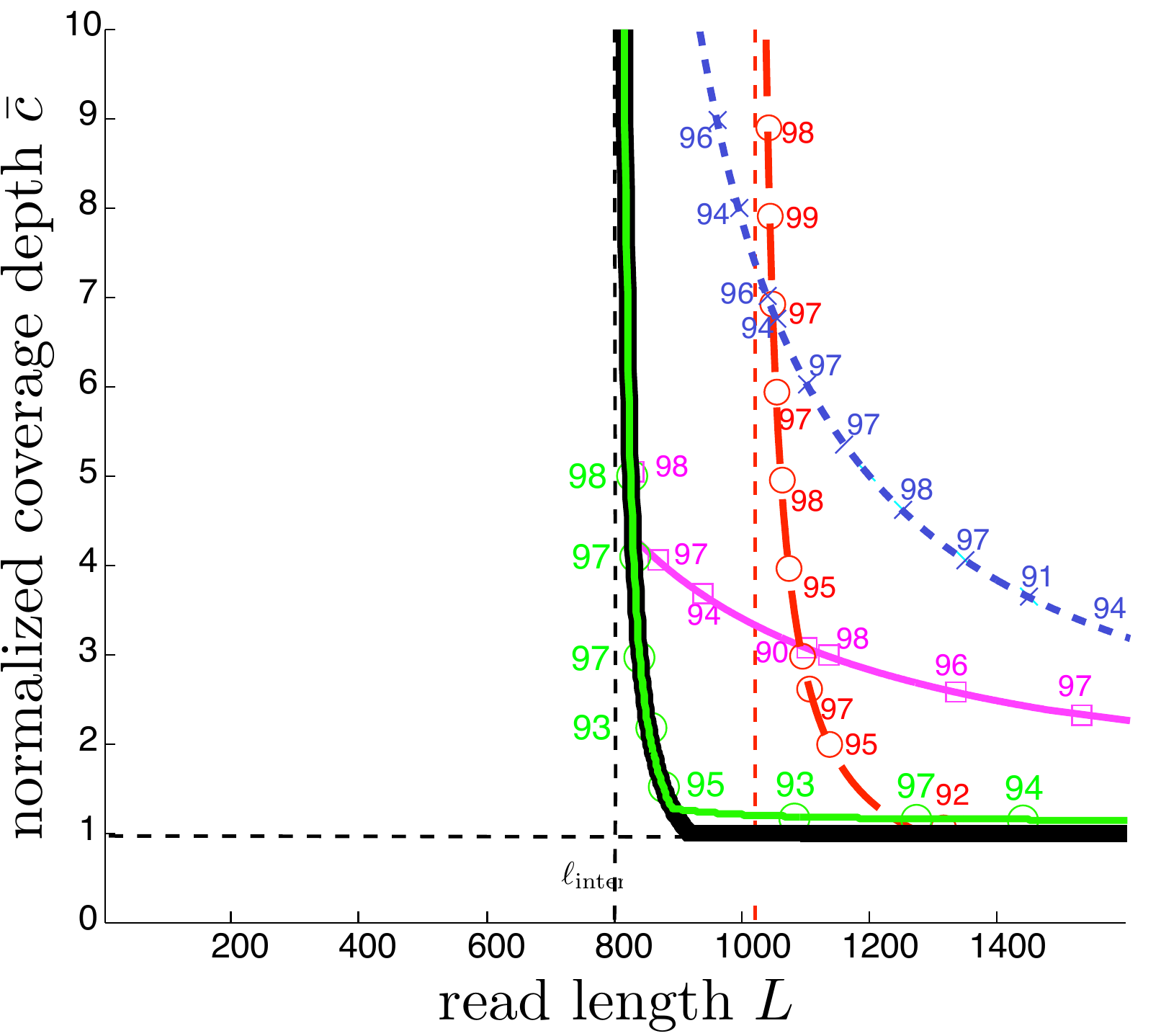}
}
\vspace{-2mm}
\caption{Simulation results for each of the GAGE reference genomes. Each simulated $(N,L)$ point is marked with the number of correct reconstructions (e.g. 93, 98, 95) on 100 simulated read sets. All four algorithms ({\GR }, {\db }, {\bdbI }, and {\bdbII}) were run on \emph{S. Aureus}, \emph{R. sphaeroides} and \hc{14}. 
Note that {\bdbII} is very close to the lower bound on all $3$ datasets. }
\label{fig:sim}
\vspace{-4mm}
\end{figure*}

In order to verify performance predictions, we implemented and ran the algorithms on simulated error-free reads from sequenced genomes. For each algorithm, we sampled $(N,L)$ points predicted to give $<5\%$ error, and recorded the number of times correct reconstruction was achieved out of $100$ trials. Fig.~\ref{fig:sim} shows results for the three GAGE reference sequences. 

We now estimate the run-time of \bdbII. The  algorithm has two phases: the $K$-mer graph formation step, and the repeat resolution step. The $K$-mer graph formation runtime can be easily bounded by $O((L-K) N K)$, assuming $O(K)$ look-up time for each of the $(L-K) N$ $K$-mers observed in reads. This step is common to all $K$-mer graph based algorithms, so previous works to decrease the practical runtime or memory requirements 
are applicable. 

The repeat resolution step depends on the repeat statistics and choice of $K$. It can be loosely bounded as 
$O \Big(  \sum_{\ell=K}^{L} L \sum_{\text{max repeats $x$}\atop\text{of length } \ell} d_x \Big).
$
The second sum is over distinct maximal repeats $x$ of length $\ell$ and $d_x$ is the number of (not necessarily maximal) copies of repeat $x$.  The bound comes from the fact that each maximal repeat of length $K< \ell < L$ is resolved via exactly one bridged {\xn}, and each such resolution requires examining at most the $L d_x$ distinct reads that contain the repeat. We note that 
$ \sum_{\ell=K}^{L} L \sum_{\text{max repeats $x$}\atop\text{of length } \ell} d_x < L   \sum_{\ell = K}^L  a_\ell\,, $ 
and the latter quantity is easily computable from our sufficient statistics. 

For our data sets, with appropriate choice of $K$, the bridging step is much simpler than the $K$-mer graph formation step: for \emph{R. sphaeroides}  we use $K=40$ to get $\sum_{\ell = K}^L L a_\ell =  412$; in contrast, $N > 22421$ for the relevant range of $L$. Similarly, for \hc{14}, using $K=300$, $\sum_{\ell = K}^L L a_\ell =  661$ while $N>733550$; for \emph{S. Aureus},  $\sum_{\ell = K}^L L a_\ell =558$ while $N>8031$.



\vspace{-2mm}
\section{Discussions and extensions}
\label{sec:conclusions}
The notion of {\em optimal shotgun assembly} is not commonly discussed in the literature.  One reason is that there is no universally agreed-upon metric of success.  Another reason is that most of the optimization-based formulations of assembly have been shown to be NP-hard, including Shortest Common Superstring \cite{Gallant80}, \cite{KM93}, De Bruijn Superwalk  \cite{PTW01}, \cite{medvedev07}, and Minimum s-Walk on the string graph \cite{Mye05}, \cite{medvedev07}. Thus, it would seem that  optimal assembly algorithms are out of the question from a computational perspective. What we show in this paper is that if the goal is complete reconstruction, then one can define a clear notion of optimality, and moreover there is a computationally efficient assembly algorithm ({\bdbII}) that is near optimal for a wide range of DNA repeat statistics. So while the reconstruction problem may well be NP-hard, typical instances of the problem seem much easier than the worst-case,  a possibility already suggested by Nagarajan and Pop \cite{NP09}. 

The {\bdbII} algorithm is near optimal in the sense that, for a wide range of repeat statistics,  it requires the minimum read length and minimum coverage depth to achieve complete reconstruction. However, since the repeat statistics of a genome to be sequenced are usually not known in advance, this minimum required read length and minimum required coverage depth may also not be known in advance.  In this context, it would be useful for the {\bdbII} algorithm to {\em validate} whether its assembly is correct. More generally, an interesting question is to seek algorithms which are not only optimal in their data requirements but also provide a measure of  confidence in their assemblies.

How realistic is the goal of complete reconstruction given current-day sequencing technologies? The minimum read lengths $\lcrit$ required for complete reconstruction on the datasets we examined  are typically on the order of $500-3000$ base pairs (bp). This is substantially longer than the reads produced by Illumina,  the current dominant sequencing technology, which produces reads of lengths 100-200bp; however, other technologies produce longer reads. PacBio reads can be as long as several thousand base pairs, and as demonstrated by \cite{KSW12}, the noise can be cleaned by Illumina reads to enable near-complete reconstruction. Thus our framework is already relevant to some of the current cutting edge technologies. To make  our framework more relevant to short-read technologies such as Illumina, an important direction is to incorporate mate-pairs in the read model, which can help to resolve long repeats with short reads. Other extensions to the basic shotgun sequencing model:

\noindent
{\bf heterogenous read lengths}: This occurs in some technologies where the read length is random (e.g. Pacbio) or when reads from multiple technologies are used.  Generalized Ukkonen's conditions and the sufficient conditions of {\bdbII} extend verbatim to this case, and only the computation of the bridging probability (\ref{e:pub}) has to be slightly modified.

\noindent
{\bf non-uniform read coverage}: Again, only the computation of the bridging probability has to be modified.  One issue of interest is to investigate whether reads are sampled less frequently from  long repeat regions. If so, our framework can quantify the performance hit.

\noindent
{\bf double strand}: DNA is double-stranded and consists of a length-$G$ sequence $\bu$ and its reverse complement $\ru$. Each read is either sampled from $\bu$ or $\ru$. This more realistic scenario can be mapped into our single-strand model by defining $\bs$ as the length-$2G$ concatenation of $\bu$ and $\ru$, transforming each read into itself and its reverse complement so that there are $2N$ reads. Generalized Ukkonen's conditions hold verbatim for this problem, and {\bdbII} can be applied, with the slight modification that instead of looking for a single Eulerian path, it should look for two Eulerian paths, one for each component of the sequence graph  after repeat-resolution.  An interesting aspect of this model is that, in addition to interleaved repeats on the single strand $\bu$,  {\em reverse complement repeats} on $\bu$ will also induce interleaved repeats on the sequence $\bs$. 





\clearpage
\bibliographystyle{amsplain}
\bibliography{DNArefs}
\clearpage

\appendix

\section{Supplementary Material}
\label{sec:plots}
%


In this supplementary material, we display the output of our pipeline for 9 datasets (in addition to \hc{19}, whose output is in the introduction, and the GAGE datasets R. sphaeroides,  S. Aureus, and \hc{14}). For each dataset we plot
$$
\log(1+a_\ell),
$$ the log of one plus the number of repeats of each length $\ell$. From the repeat statistics $a_m$, $b_{m,n}$, and $c_m$, we produce a feasibility plot. The thick black line denotes the lower bound on feasible $N,L$, and the green line is the performance achieved by {\bdbII}.
%



\begin{figure*}[htb]
\begin{centering}
\subfloat{
\includegraphics[height=2.5in]{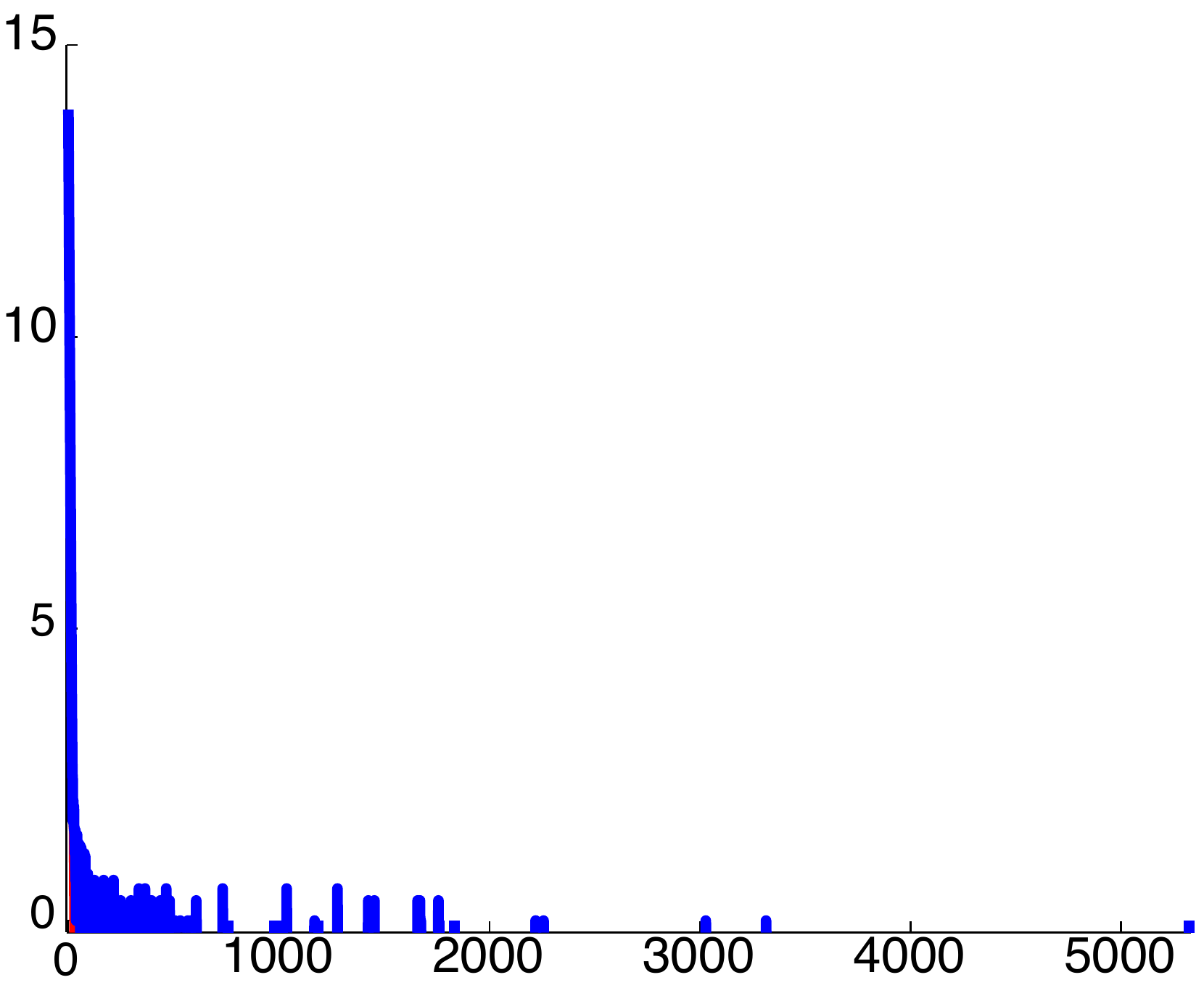}
}
\subfloat{
\includegraphics[height=2.5in]{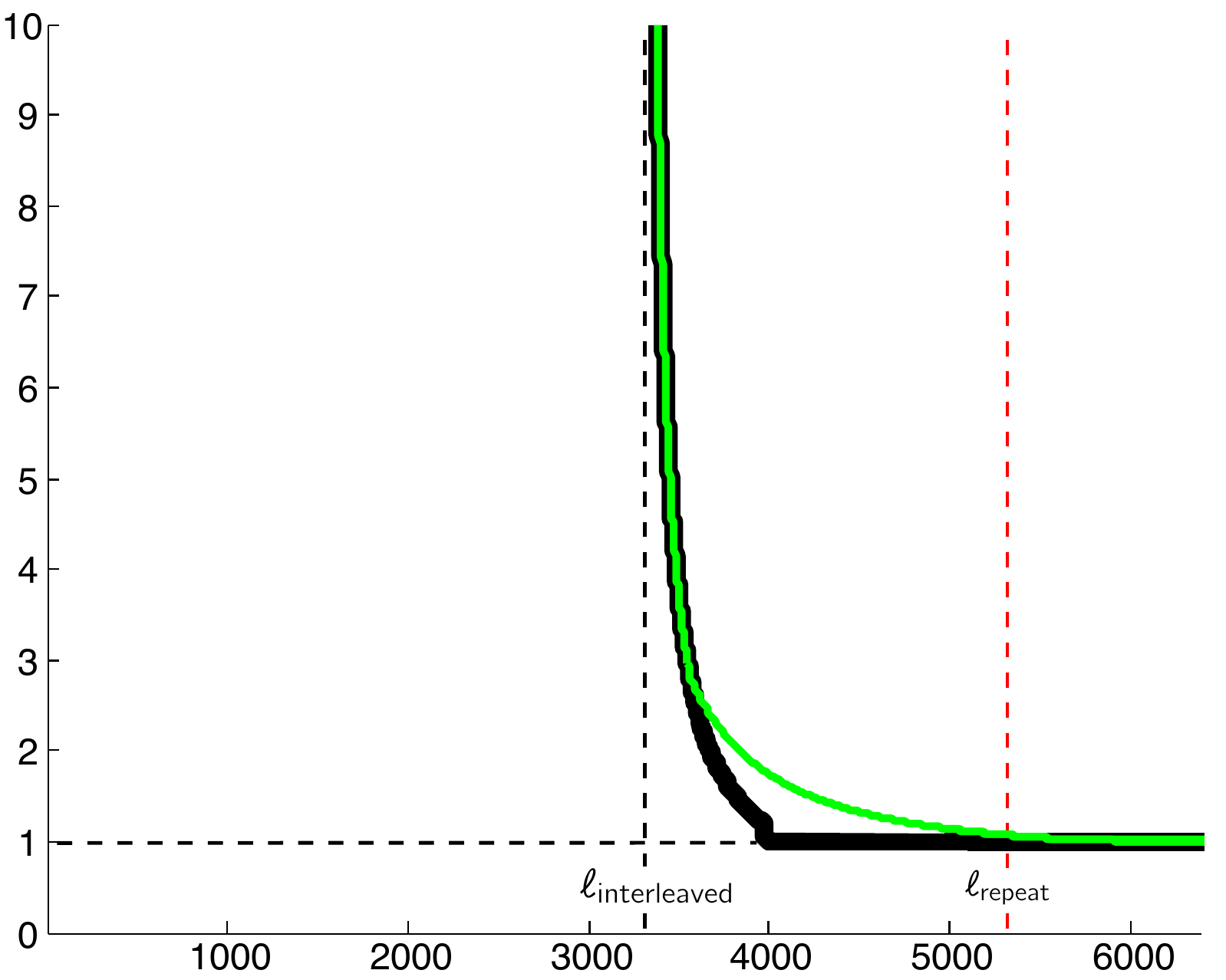}
}
\end{centering}
\caption{Lactofidus. $G=2,078,001$, $\Ltri = 3027$, $\Lint=3313$, $\Lrep=5321$.}
\vspace{-0.3cm}
\end{figure*}

\begin{figure*}[h]
\begin{centering}
\subfloat{
\includegraphics[height=2.5in]{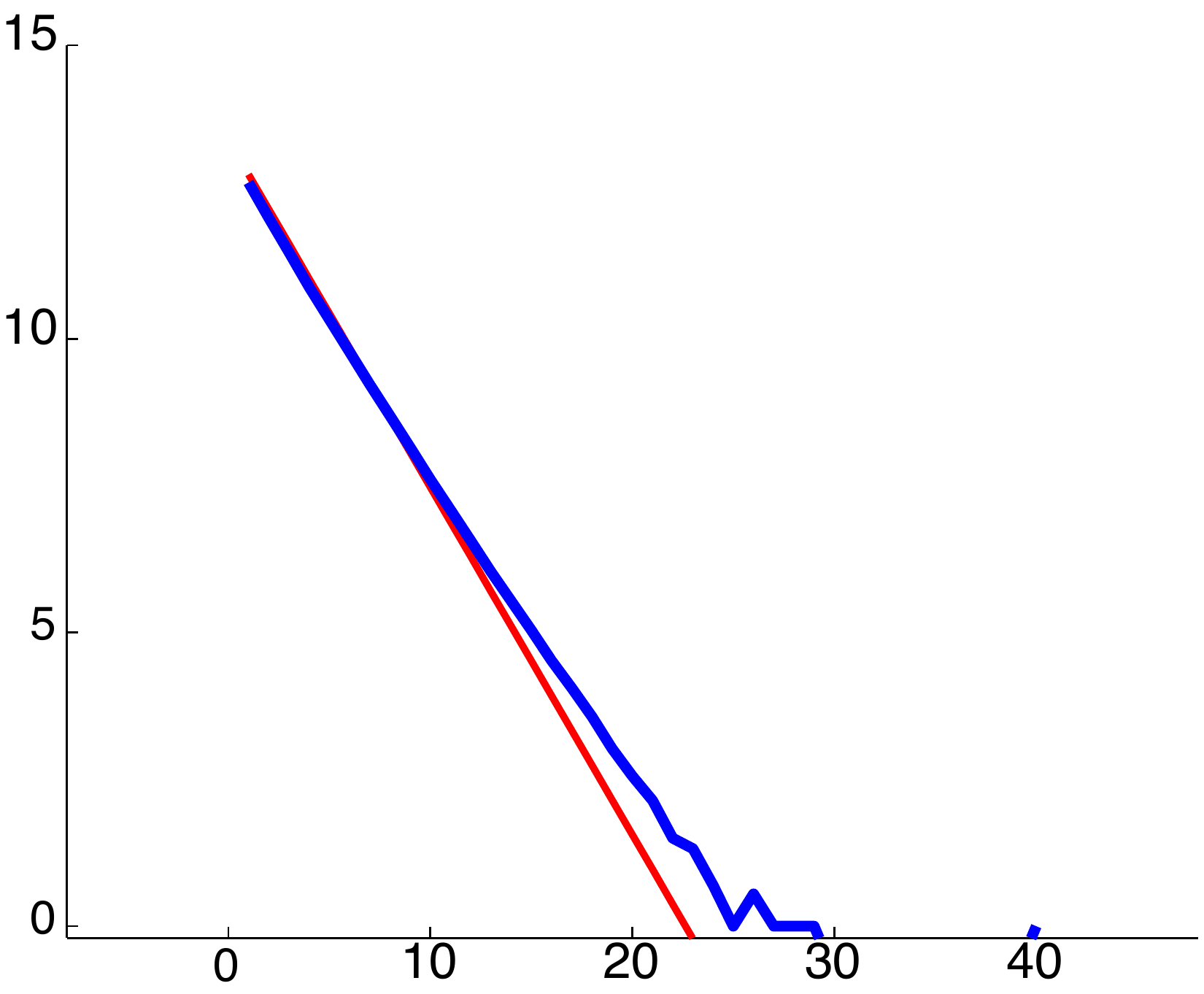}
}
\subfloat{
\includegraphics[height=2.5in]{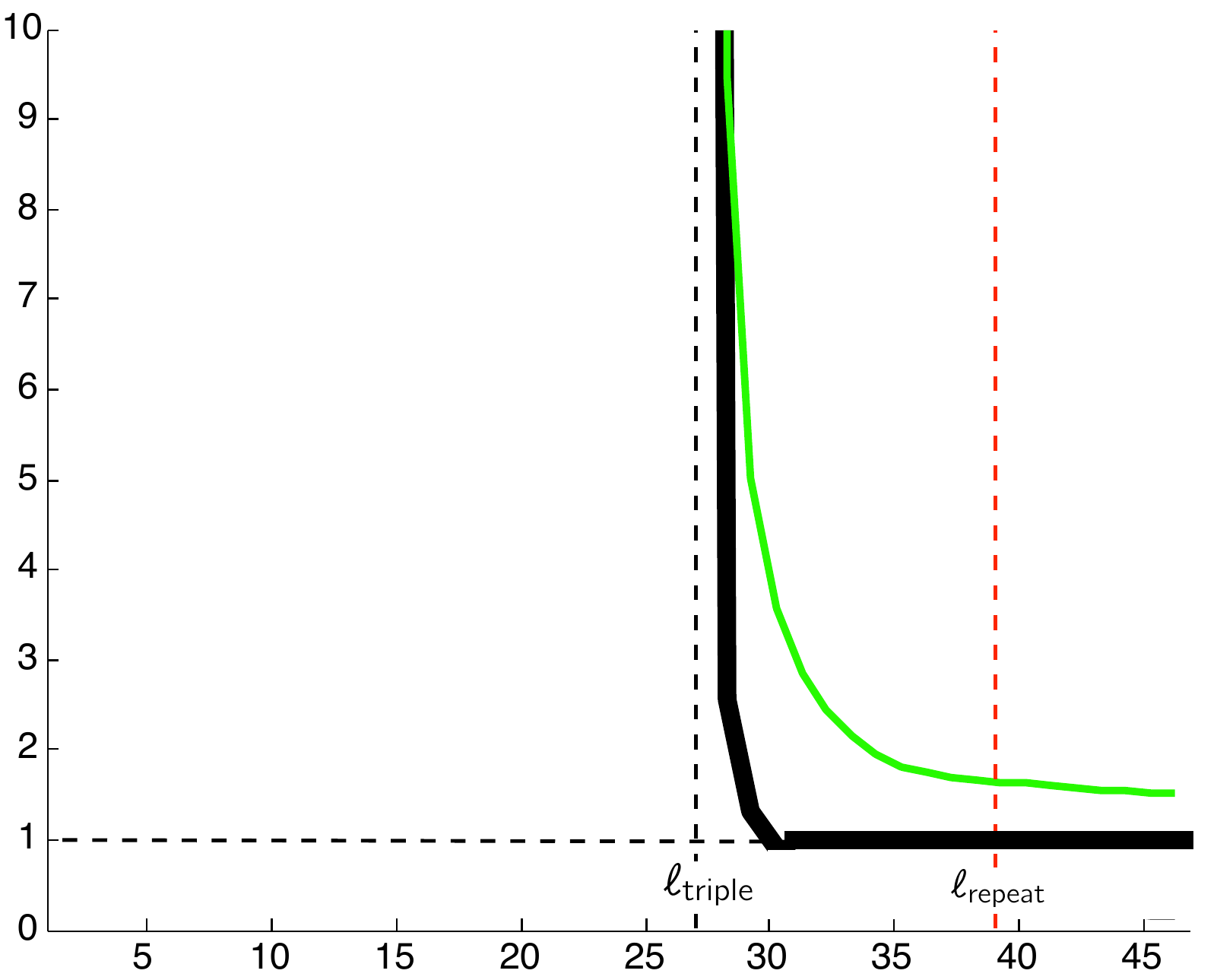}
}
\end{centering}
\caption{Buchnera. $G=642,122$, $\Ltri = 27$, $\Lint=23$, $\Lrep=39$.}
\vspace{-0.3cm}
\end{figure*}

\begin{figure*}[h]
\begin{centering}
\subfloat{
\includegraphics[height=2.5in]{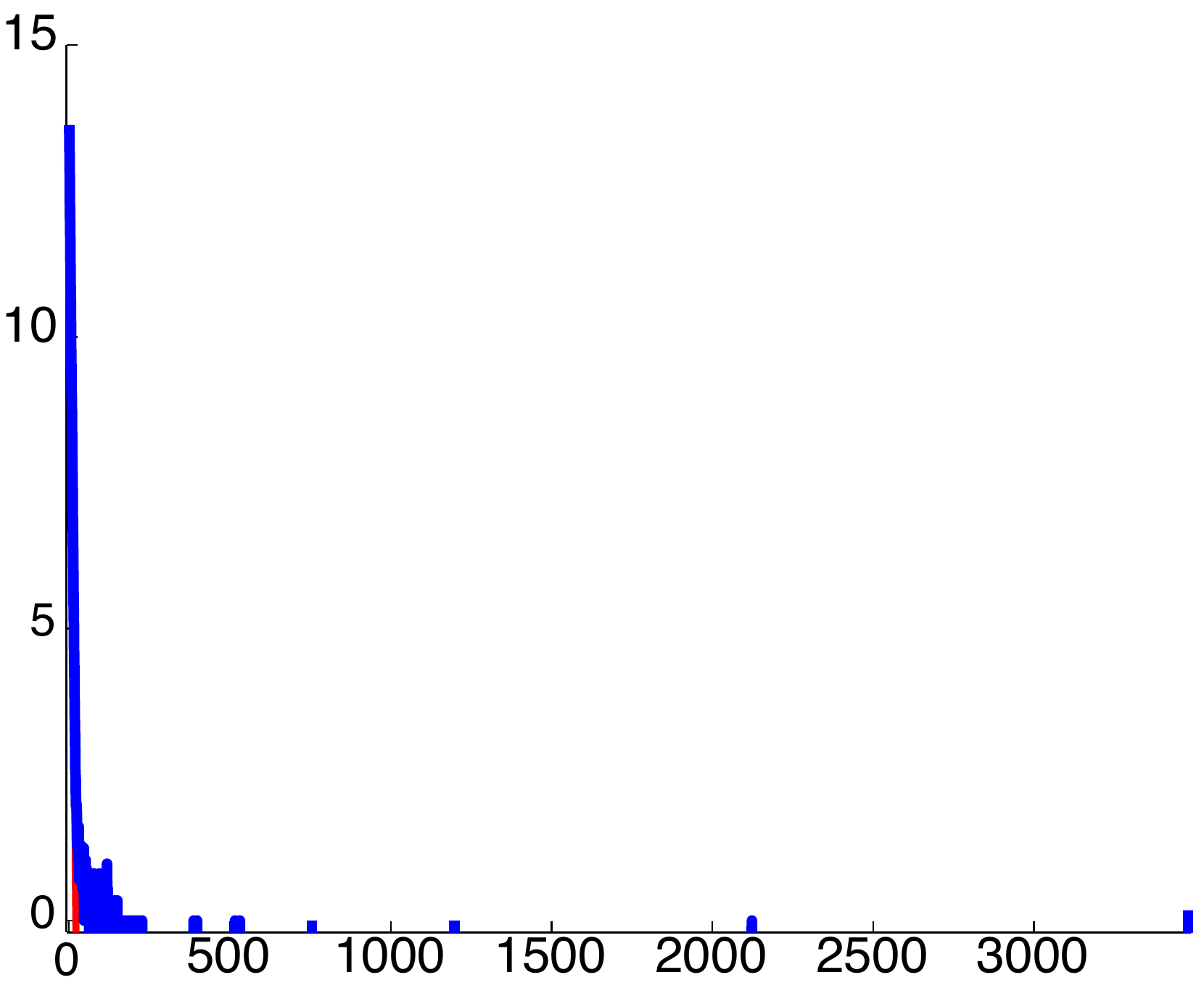}
}
\subfloat{
\includegraphics[height=2.5in]{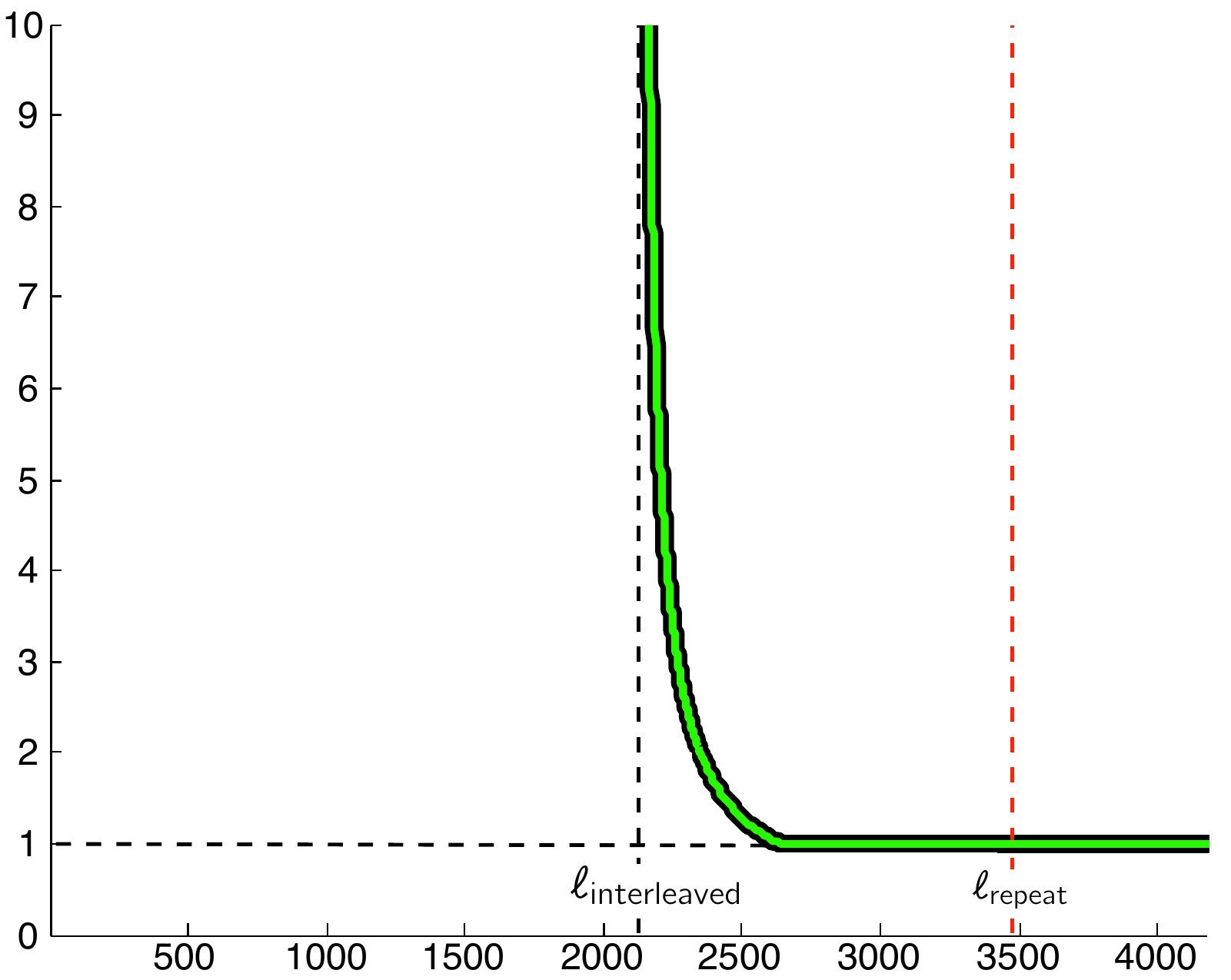}
}
\end{centering}
\caption{Heli51. $G=1,589,954$, $\Ltri =219$, $\Lint = 2122$, $\Lrep = 3478$.}
\vspace{-0.3cm}
\end{figure*}

%

\begin{figure*}[h]
\begin{centering}
\subfloat{
\includegraphics[height=2.5in]{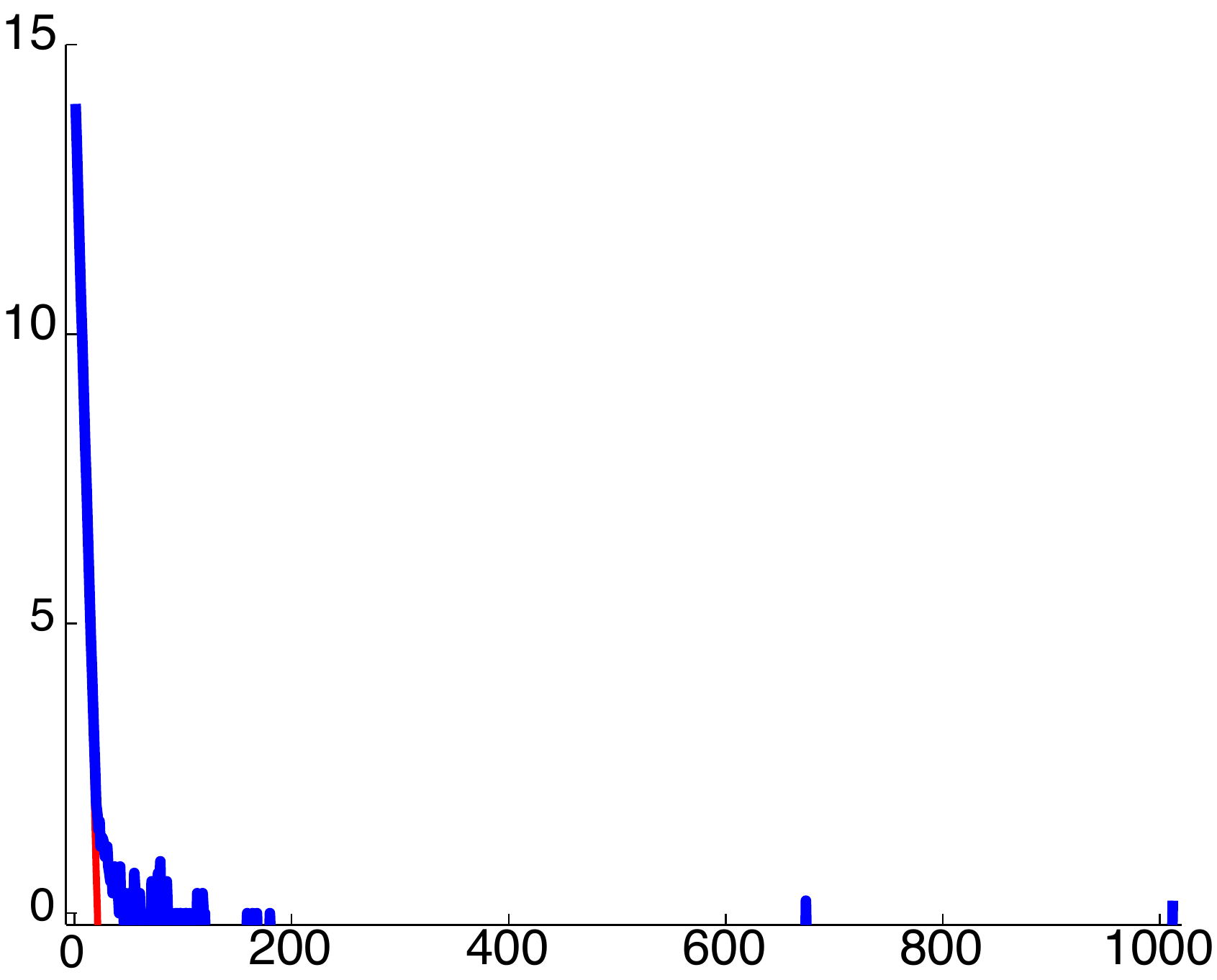}
}
\subfloat{
\includegraphics[height=2.5in]{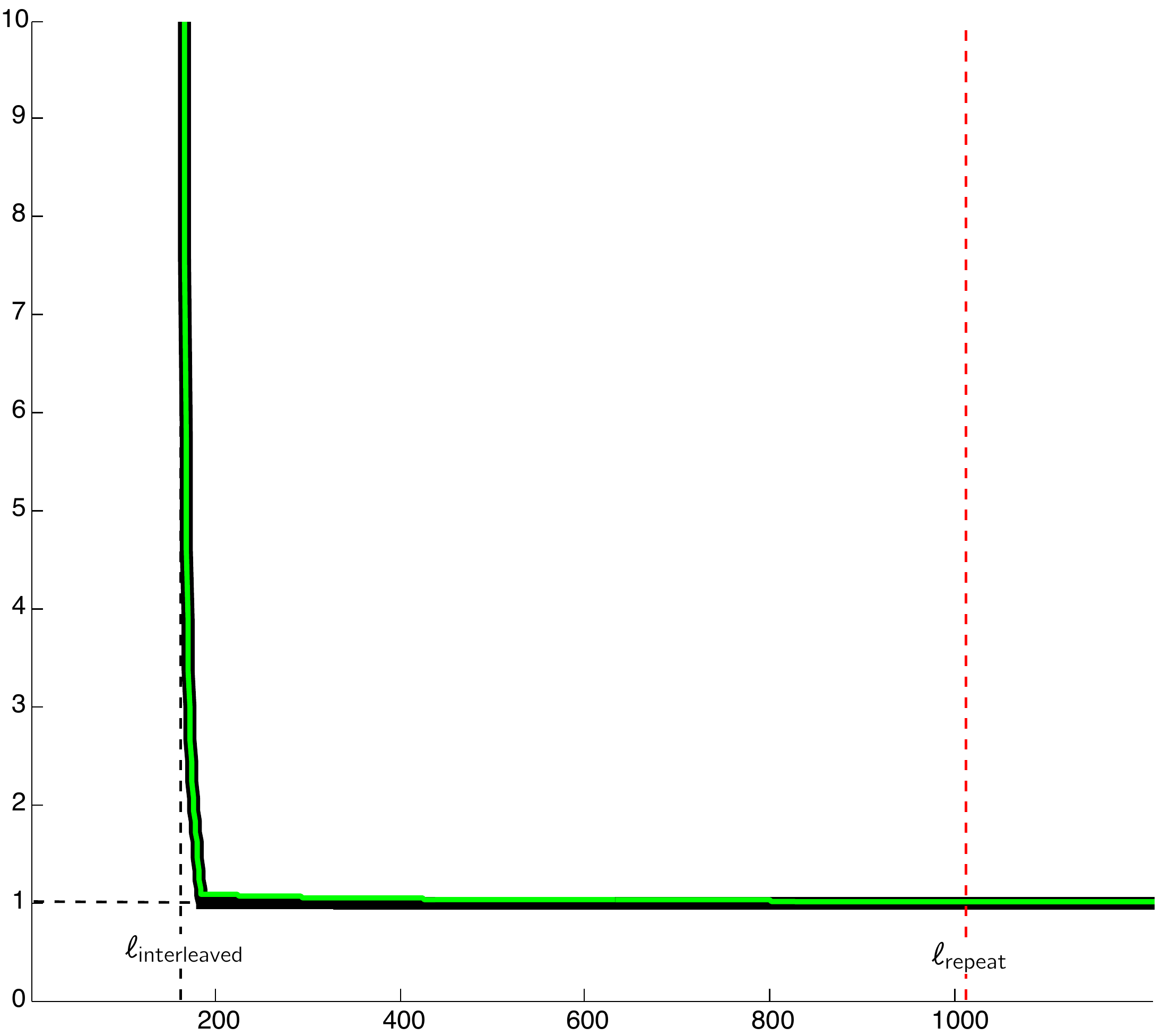}
}
\end{centering}
\caption{Salmonella. $G=2,215,568$, $\Ltri =112$, $\Lint =163$, $\Lrep =1011$.}
\vspace{-0.3cm}
\end{figure*}

\begin{figure*}[h]
\begin{centering}
\subfloat{
\includegraphics[height=2.5in]{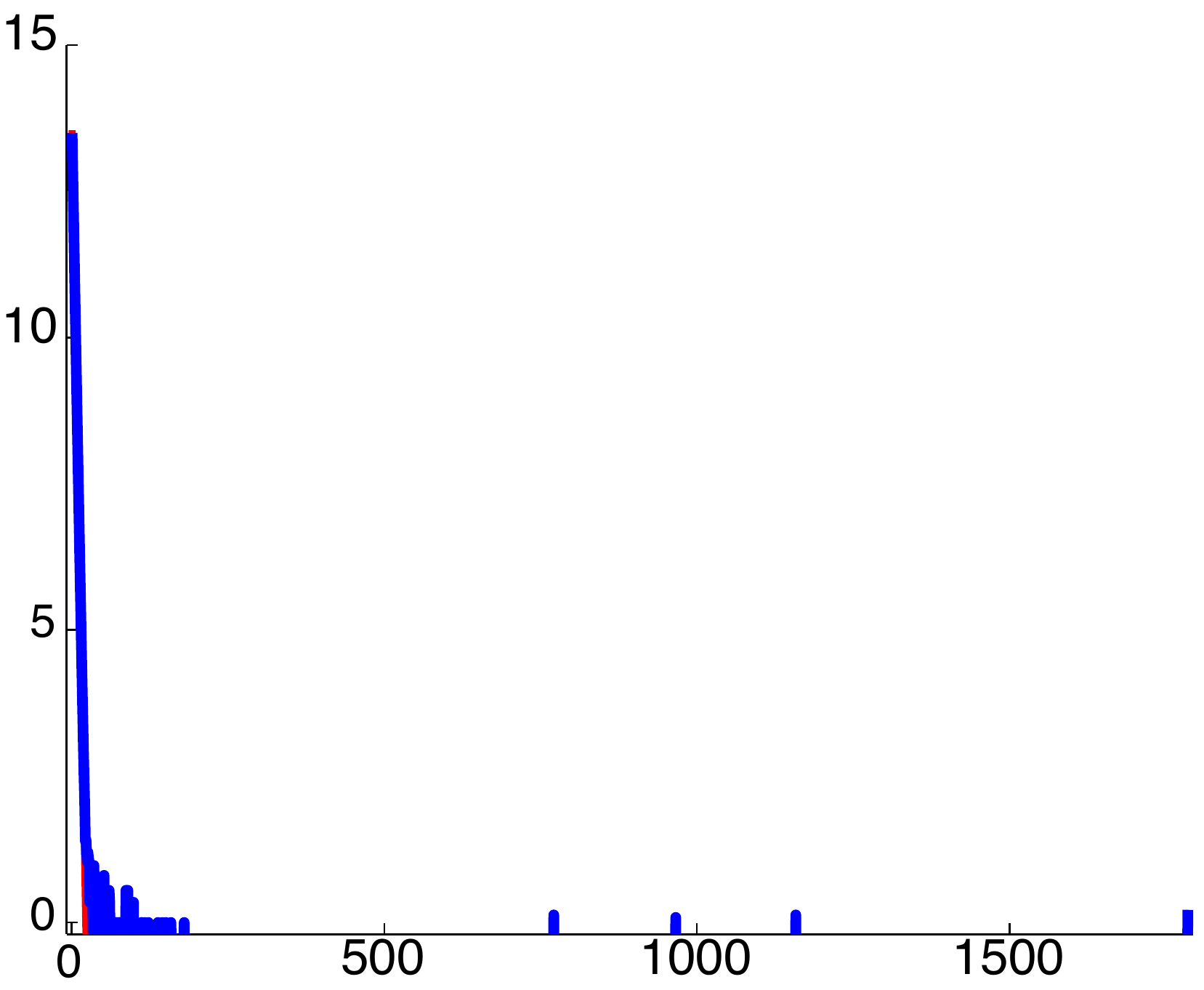}
}
\subfloat{
\includegraphics[height=2.5in]{plots/NICEPLOTS/PmarinusPlot.pdf}
}
\end{centering}
\caption{Perkinsus marinus. $G=1,440,372$, $\Ltri =770$, $\Lint =92$, $\Lrep =1784$.}
\vspace{-0.3cm}
\end{figure*}

\begin{figure*}[h]
\begin{centering}
\subfloat{
\includegraphics[height=2.5in]{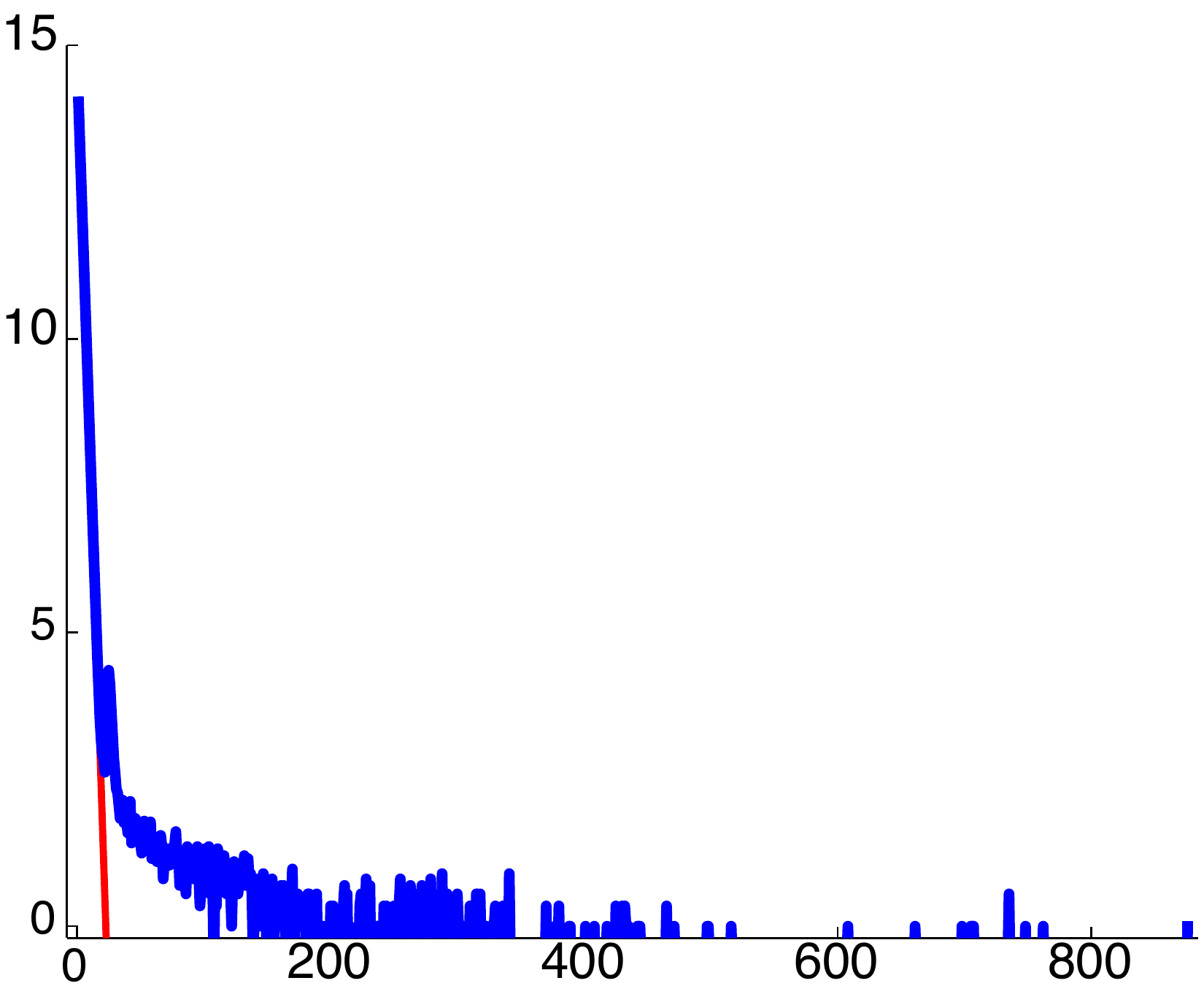}
}
\subfloat{
\includegraphics[height=2.5in]{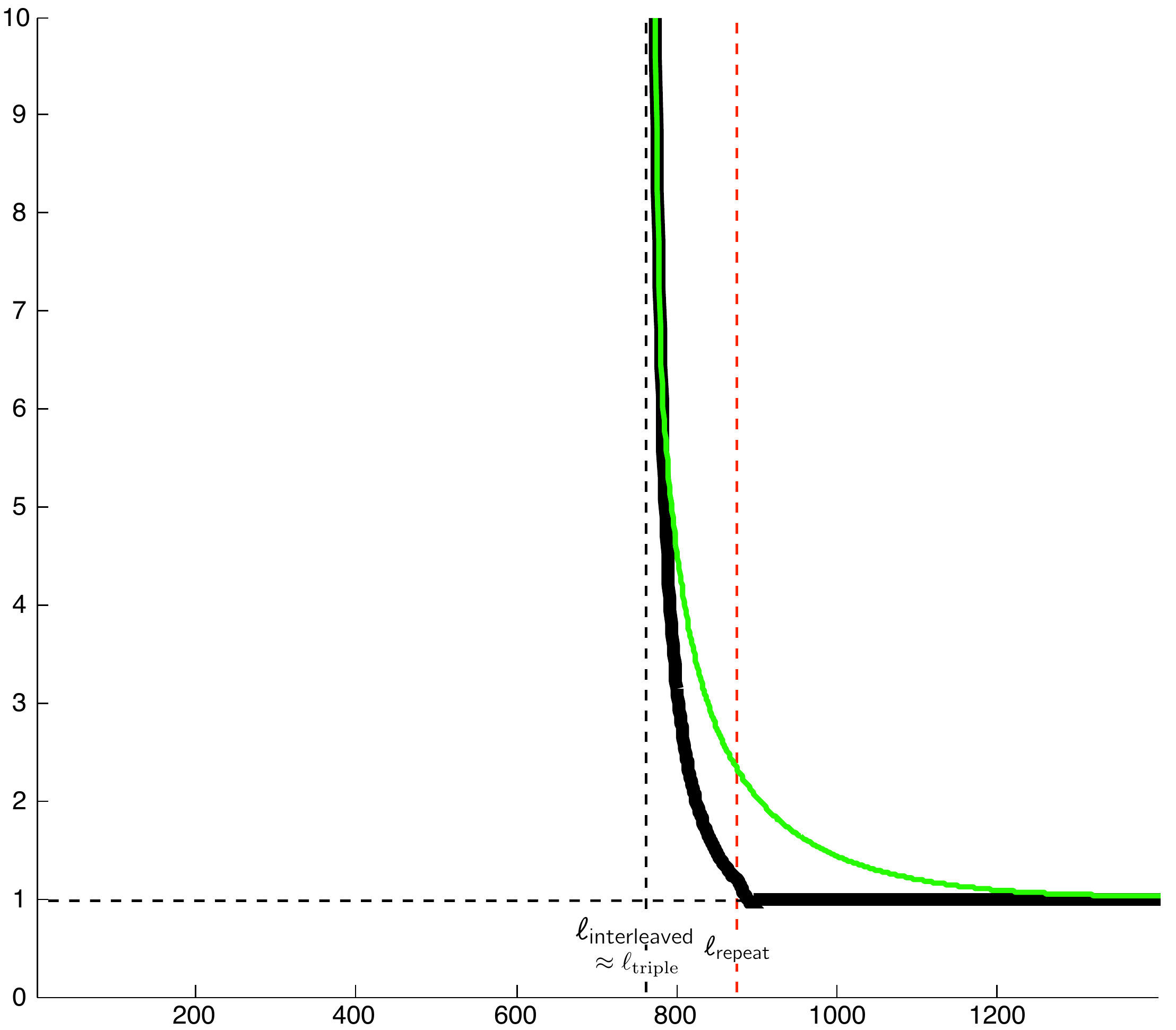}
}
\end{centering}
\caption{Sulfolobus islandicus. $G=2,655,198$, $\Ltri = 734$, $\Lint = 761$, $\Lrep = 875$.}
\vspace{-0.3cm}
\end{figure*}

%

\begin{figure*}[h]
\begin{centering}
\subfloat{
\includegraphics[height=2.5in]{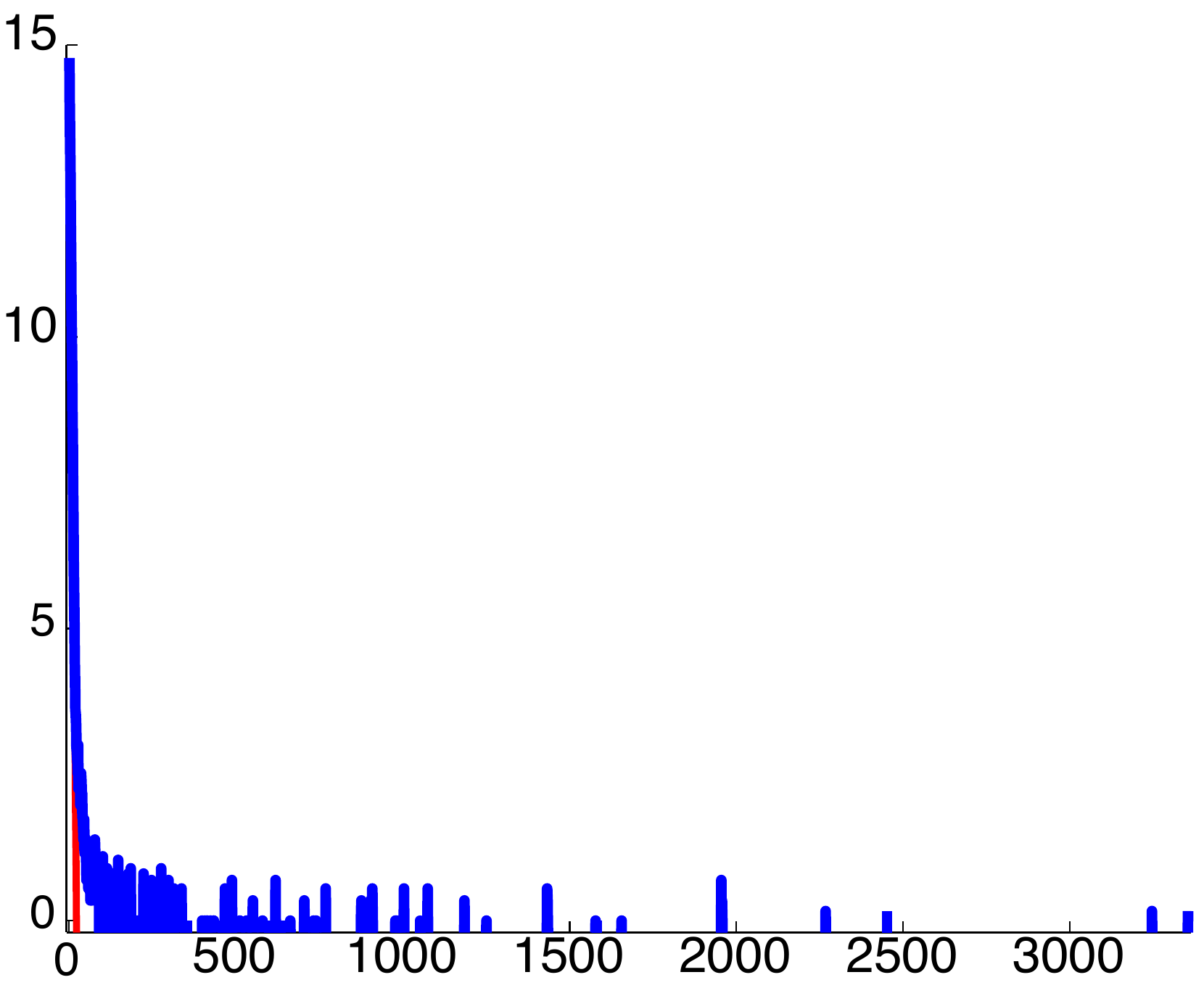}
}
\subfloat{
\includegraphics[height=2.5in]{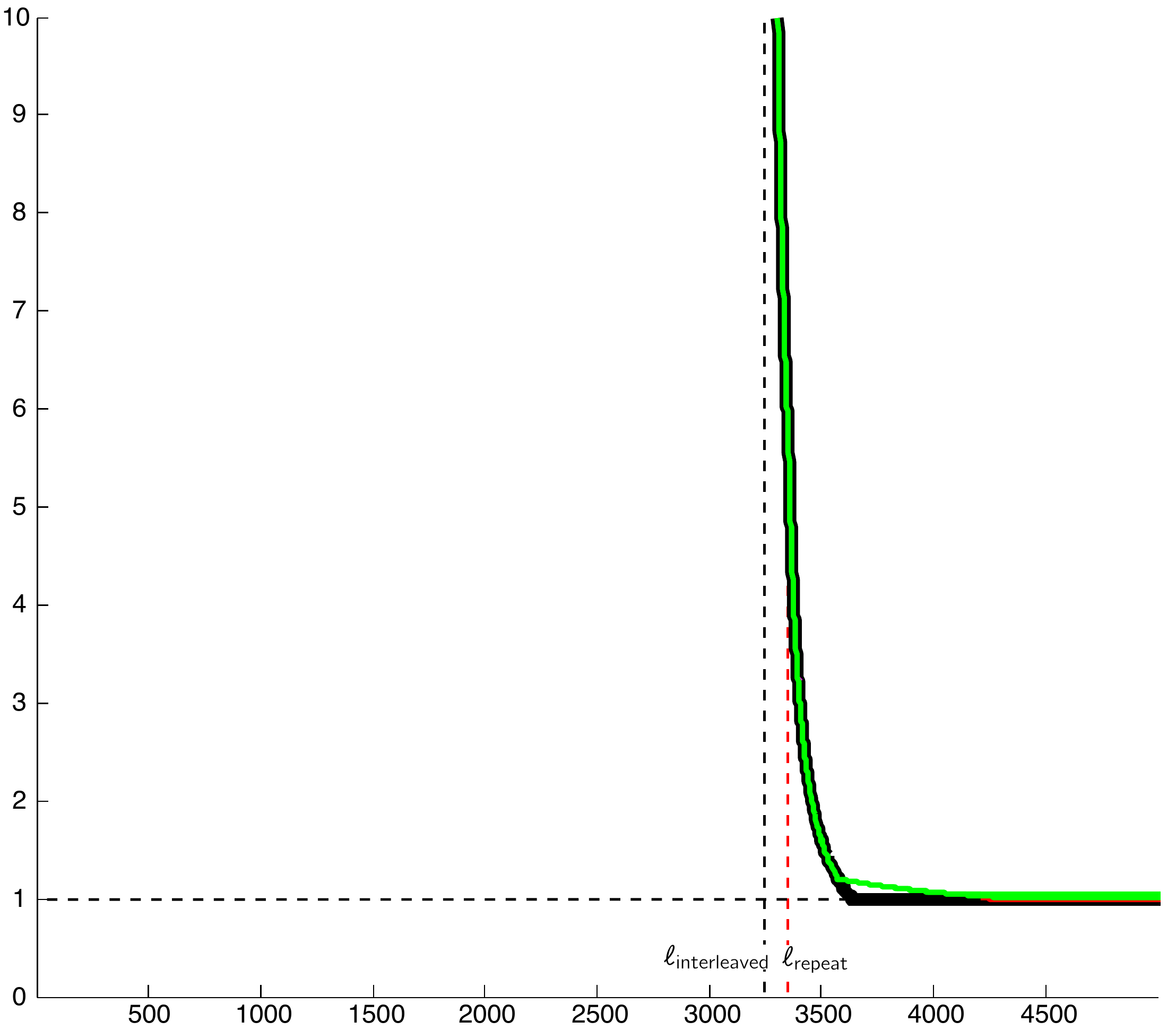}
}
\end{centering}
\caption{Ecoli536. $G=4,938,920$, $\Ltri = 2267$, $\Lint = 3245$, $\Lrep = 3353$.}
\vspace{-0.3cm}
\end{figure*}

\begin{figure*}[h]
\begin{centering}
\subfloat{
\includegraphics[height=2.5in]{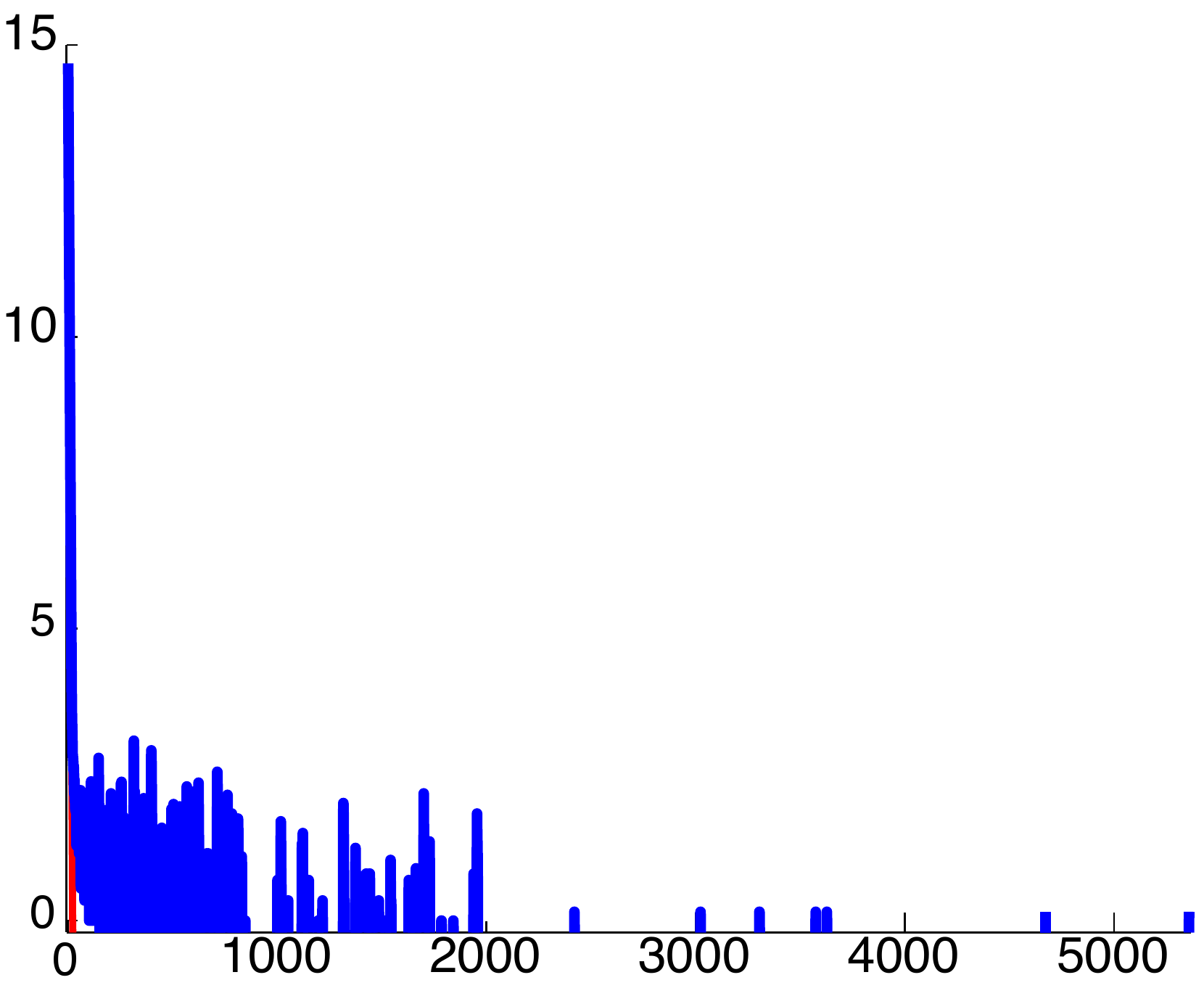}
}
\subfloat{
\includegraphics[height=2.5in]{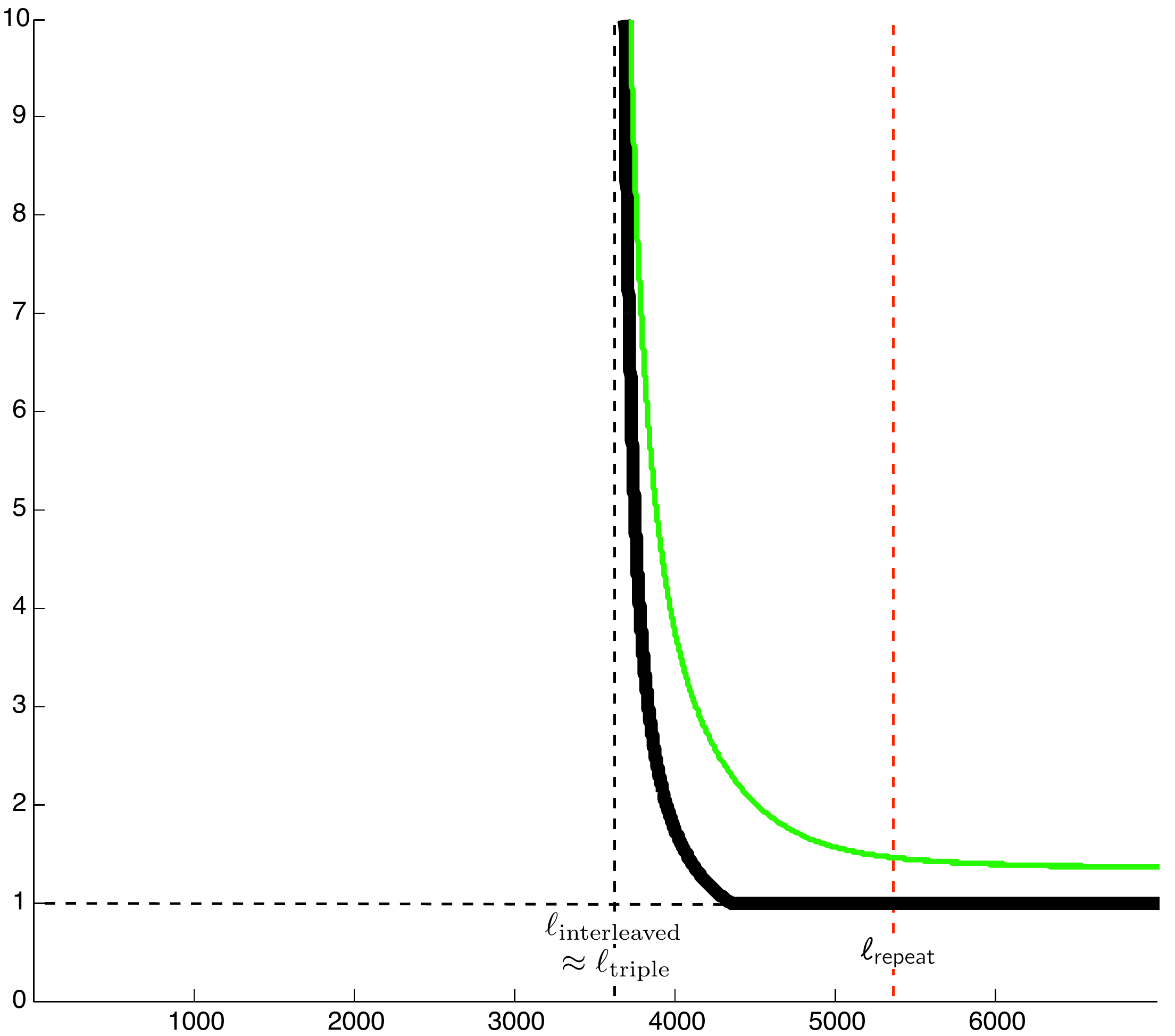}
}
\end{centering}
\caption{Yesnina. $G=4,504,254$, $\Ltri = 3573$, $\Lint = 3627$, $\Lrep = 5358$.}
\vspace{-0.3cm}
\end{figure*}

%

\clearpage

\section{Lower bounds on coverage depth}\label{sec:AppLower}

The lower bounds are based on a generalization of Ukkonen's condition to shotgun sequencing, as described in Theorem~\ref{t:Ukkonen}. The proof of Theorem~\ref{t:Ukkonen} follows by a straightforward modification to the argument in \cite{Ukk92} and is omitted here.
\begin{reptheorem}{t:Ukkonen}
Given a DNA sequence ${\bf s}$ and a set of reads, if there is a pair of  interleaved repeats or a triple repeat whose copies are all unbridged, then there is another sequence ${\bf s'}$ of the same length under which the likelihood of observing the reads is the same.
\end{reptheorem}

\subsection{Lower bound due to interleaved repeats}
In this section we derive a necessary condition on $N$ and $L$ in order that the probability of correct reconstruction be at least $1-\eps$. 

Recall that a pair of repeats, one at positions $t_1,t_3$ with $t_1<t_3$ and the second at positions $t_2,t_4$ with $t_2<t_4$, is \emph{interleaved} if $t_1<t_2<t_3<t_4$ or $t_2<t_1<t_4<t_3$. 
From the DNA we may extract a (symmetric) matrix of interleaved repeat statistics
$b_{mn}$, the number of pairs of interleaved repeats of lengths $m$ and $n$. 

We proceed by fixing both $N$ and $L$ and checking whether or not unbridged interleaved repeats occur with probability higher than $\eps$. 
We will break up repeats into 2 categories: repeats of length at least $L-1$ (these are always unbridged), and repeats of length less than $L-1$ (these are sometimes unbridged). We assume that $L>\Lint+1$, or equivalently $b_{ij}=0$ for all $i,j\geq L-1$, since otherwise there are (with certainty) unbridged interleaved repeats and Ukkonen's condition is violated.

First, we estimate the probability of error due to interleaved repeats of lengths $i<L-1$ and $j\geq L-1$. The repeat of length $j$ is too long to be bridged, so an error occurs if the repeat of length $i$ is unbridged. For a repeat, as long as the two copies' locations are not too nearby\footnote{More precisely, for the two copies of a a repeat  of length $\ell$ to be bridged independently requires that no single read can bridge them both. This means their locations $t$ and $t+d$ must have separation $d\geq L-\ell-2$. }, each copy is bridged independently and hence the probability that both copies of the repeat of length $i$ are unbridged is $\pub_i = e^{-2\frac NG(L-i-1)}$. (Recall that a repeat is unbridged if \emph{both} copies are unbridged.)

A union bound estimate\footnote{The union bound on probabilities gives an \emph{upper bound}, so its use here is only an approximation. To get a rigorous lower bound we can use the inclusion-exclusion principle, but the difference in the two computations is negligible for the data we observed. For ease of exposition we opt to present the simpler union bound estimate.}  gives a probability of error
\begin{equation} \label{e:intLowerBound1}
\Pe\approx 
\frac12 \sum_{m< L-1\atop n\geq L-1}b_{mn}e^{-2\lam(L-m-1)}\,.
\end{equation}
Requiring the error probability to be less than $\eps$ and solving for $L$ gives the necessary condition
\begin{equation}\label{e:LintBound1Appendix}
L\geq \frac1{2\lam}\log\frac{\gamma_1}{2\eps}=\frac G{2N}\log\frac{\gamma_1}{2\eps} \,,
\end{equation}
where $\gamma_1:=\sum_{m< L-1\atop n\geq L-1}b_{mn} e^{2(N/G)(m+1)}$ is a simple function of the interleaved repeat statistic $b_{mn}$.

We now estimate the probability of error due to interleaved repeat pairs in which both repeats are shorter than $L-1$. In this case only one repeat of each interleaved repeat pair must be bridged. Again a union bound estimate gives
$$
\Pe\approx \frac12 \sum_{m,n<L-1} b_{mn} e^{-2\lam(L-m-1)}e^{-2\lam(L-n-1)}\,.
$$
Requiring the error probability to be less than $\eps$ gives the necessary condition
\begin{equation}\label{e:LintBound2Appendix}
L\geq \frac1{4\lam}\log\frac{\gamma_2}{2\eps}=  \frac G{4N}\log\frac{\gamma_2}{2\eps}\,,
\end{equation}
where $\gamma_2:=\sum_{m,n< L-1}b_{mn} e^{2(N/G)(m+n+2)}$ and similarly to $\gam_1$ is computed from $b_{mn}$.

\subsection{Lower bound due to triple repeats} \label{sec:TripleRepeats}
We translate the generalized Ukkonen's condition prohibiting unbridged triple repeats into a condition on $L$ and $N$. Let $c_m$ denote the number of triple repeats of length $m$. Then a union bound estimate gives
\begin{equation}
\P(\EE)\approx \frac12 \sum_m c_m e^{-3\lam(L-m-1)}\,.
\end{equation}
Requiring $\P(\EE)\leq \eps$ and solving for $L$ gives 
\begin{equation}\label{e:LtripBoundAppendix}
L\geq \frac{1}{3\lam}\log\frac{\gam_3}{2\eps}= \frac{G}{3N}\log\frac{\gam_3}{2\eps}\,,
\end{equation}
where $\gam_3 := \sum_m c_m e^{3(N/G) (m+1)}$.

\begin{remark}
As discussed here and in Section~\ref{sec:Lower}, if the DNA sequence is not covered by the reads or there are unbridged interleaved or triple repeats, then reconstruction is not possible. But there is another situation which must be ruled out. Without knowing its length a priori, it is impossible to know how many copies of the DNA sequence are actually present: if the sequence $\s$ to be assembled consists of multiple concatenated copies of a shorter sequence, rather than just one copy, the probability of observing any set of reads will be the same. Since it is unlikely that a true DNA sequence will consist of the same sequence repeated multiple times, we assume this is not the case throughout the paper. Equivalently, if $\s$ does consist of multiple concatenated copies of a shorter sequence, we are content to reconstruct a single  copy. If available, knowledge of the approximate length of $\s$ would then allow to reconstruct.
\end{remark}

\section{Proofs for algorithms}
\subsection{Proof of Theorem~\ref{thm:GreedyTheorem} (\GR)} 
\label{sec:greedyApp}
The greedy algorithm's underlying data structure is the overlap graph, where each node represents a read and each (directed) edge $(\bx,\by)$ is labeled with the overlap $\ov(\bx,\by)$ (defined as the the length of the shared prefix/suffix) between the incident nodes' reads.  
For a node $\bv$, the in-degree [out-degree] is the number of edges in the graph directed towards [away from] $\bv$.
The greedy algorithm is described as follows. 

\label{sec:appGreedy}
\begin{algorithm}[H]
\caption{{\sc Greedy}. Input: reads $\RR$. Output: sequence $\sh$. }
\label{alg1}
\noindent
1. For each read with sequence $\x$, form a node with label $\x$.  \\
\emph{Greedy steps 2-3:}\\
2. Consider all pairs of nodes $\x_1,\x_2$ in $\G$ satisfying $\dout(\x_1)=\din(\x_2)=0$, and add an edge $(\x_1,\x_2)$ with largest value $\ov(\x_1,\x_2)$. \\
3. Repeat Step 2 until no candidate pair of nodes remains.\\
\emph{Finishing step:}\\
4. Output the sequence corresponding to the unique cycle in $\G$. 
\end{algorithm}

\begin{reptheorem}{thm:GreedyTheorem}
Given a sequence ${\bf s}$ and a set of reads, \GR { }returns ${\bf s}$ if every repeat is bridged. 
\end{reptheorem}
\begin{proof}We prove the contrapositive. Suppose {\GR } makes its first error in merging reads $\br_i$ and $\br_j$ with overlap $\ov(\br_i,\br_j)=\ell$.  
Now, if $\br_j$ is the successor to $\br_i$, then the error is due to incorrectly aligning the reads; the other case is that $\br_j$ is not the successor of $\br_i$. In the first case, the subsequence $\s_{t_j}^\ell$ is repeated at location $\s_{t_i+L-\ell}^\ell$, and no read bridges either repeat copy.

In the second case, there is a repeat $\s_{t_j}^\ell = \s_{t_i+L-\ell}^\ell$. If $\s_{t_i+L-\ell}^\ell$ is bridged by some read $\br_k$, then $\br_i$ has overlap at least $\ell+1$ with $\br_k$, implying that read $\br_i$ has already found its successor before step $\ell$ (either $\br_k$ or some other read with even higher overlap). A similar argument shows that $\s_{t_j}^\ell$ cannot be bridged, hence there is an unbridged repeat.
\end{proof}

\subsection{Proofs for $K$-mer algorithms}
\label{sec:SBHAppendix}

\subsubsection{Background}
We give some mathematical background leading to the proof of Theorem~\ref{t:SBH_no_multiplicities} (restated below). 

\begin{lemma}\label{l:cycleLemma}
Fix an arbitrary $K$ and form the $K$-mer graph from the $(K+1)$-spectrum $\calS_{K+1}$.
The sequence $\s$ corresponds to a unique cycle $\CC(\s)$ traversing each edge at least once.
\end{lemma}
To prove the lemma, note that all $(K+1)$-mers in $\s$ correspond to edges and adjacent $(K+1)$-mers in $\s$ are represented by adjacent edges. An induction argument shows that $\s$ corresponds to a cycle. The cycle traverses all the edges, since each edge represents a unique $(K+1)$-mer in $\s$. 

In both SBH and shotgun sequencing the number of times each edge $e$ is traversed by $\CC(\s)$ (henceforth called the \emph{multiplicity} of $e$) is unknown a priori, and finding this number is part of the reconstruction task. 
Repeated $(K+1)$-mers in $\s$ correspond to edges in the $K$-mer graph traversed more than once by $\CC(\s)$, i.e. having multiplicity greater than one. In order to estimate the multiplicity, previous works seek a solution to the so-called Chinese Postman Problem (CPP), in which the goal is to find a cycle of the shortest total length traversing every edge in the graph (see e.g. \cite{Pev89}, \cite{IW95}, \cite{PTW01}, \cite{medvedev2009maximum}). It is not obvious under what conditions the CPP solution \emph{correctly} assigns multiplicities in agreement with $\CC(\s)$. For our purposes, as we will see in Theorem~\ref{t:SBH_no_multiplicities}, the multiplicity estimation problem can be sidestepped (thereby avoiding solving CPP) through a modification to the $K$-mer graph. 

Ignoring the issue of edge multiplicities for a moment, Pevzner \cite{Pev95} showed for the SBH model that if the edge multiplicities are known with multiple copies of each edge included according to the multiplicities, and moreover Ukkonen's condition is satisfied, then there is a \emph{unique Eulerian cycle} in the $K$-mer graph and the Eulerian cycle corresponds to the original sequence. (An Eulerian cycle is a cycle traversing each edge exactly once.) Pevzner's algorithm is thus to find an Eulerian cycle and read off the corresponding sequence. Both steps can be done efficiently. 
\begin{lemma}[Pevzner \cite{Pev95}]\label{l:Pev95}
In the SBH setting, if the edge multiplicities are known, then there is a unique Eulerian cycle in the $K$-mer graph with $K=L-1$ if and only if there are no unbridged interleaved repeats or unbridged triple repeats.
\end{lemma}



Most practical algorithms (e.g. \cite{IW95}, \cite{Maccallum:2009qy}, \cite{Zerbino:2008fj}) condense unambiguous paths (called unitigs by Myers \cite{Mye00} in a slightly different setting) for computational efficiency. The more significant benefit for us, as shown in Theorem~\ref{t:SBH_no_multiplicities}, is that if Ukkonen's condition is satisfied then condensing the graph obviates the need to estimate multiplicities. Condensing a $K$-mer graph results in a graph of the following type.
\begin{definition}[Sequence graph] \label{def:seqGraph}
A \emph{sequence graph} is a graph in which each node is labeled with a subsequence, and edges $(\bu,\bv)$ are labeled with an \emph{overlap} $\olap\bu\bv$ such the subsequences $\bu$ and $\bv$ overlap by $\olap\bu\bv$ (the overlap is not necessarily maximal). In other words, an edge label $\olap\bu\bv$ on $e=(\bu,\bv)$ indicates that the $\olap\bu\bv$-length suffix of $\bu$ is equal to the $\olap\bu\bv$-length prefix of $\bv$. 
\end{definition}
The sequence graph generalizes both the overlap graph used by {\GR } in Section~\ref{sec:Greedy} (nodes correspond to reads, and edge overlaps are \emph{maximal} overlaps) as well as the $K$-mer algorithms discussed in this section (nodes correspond to $K$-mers, and edge overlaps are $K-1$).

In order to speak concisely about concatenated sequences in the sequence graph, we extend the notation $\s_t^\ell$ (denoting the length-$\ell$ subsequence of the DNA sequence ${\bf s}$ starting at position $t$) which was introduced in Section~\ref{sec:Ukkonen}; we abuse notation slightly, and write $\s_t^\text{end}$ to indicate the subsequence of ${\bf s}$ starting at position $t$ and having length so that its end coincides with the end of ${\bf s}$.

We will perform two basic operations on the sequence graph. For an edge $e=(\bu,\bv)$ with overlap $\olap\bu\bv$, \emph{merging} $\bu$ and $\bv$ along $e$ produces the concatenation $\bu_1^\text{end} \bv_{\olap\bu\bv+1}^\text{end}$. \emph{Contracting} an edge $e=(\bu,\bv)$ entails two steps (c.f. Fig.~\ref{fig:condensing}): first, merging $\bu$ and $\bv$ along $e$ to form a new node $\bw=\bu_1^\text{end} \bv_{\olap\bu\bv+1}^\text{end}$, and, second, edges to $\bu$ are replaced with edges to $\bw$, and edges from $\bv$ are replaced by edges from $\bw$. We will only contract edges $(\bu,\bv)$ with $\dout(\bu)=\din(\bv)=1$.

The condensed graph is defined next.  
\begin{definition}[Condensed sequence graph] \label{d:condensed}
The \emph{condensed sequence graph} replaces unambiguous paths by single nodes. Concretely, any edge $e=(u,v)$ with $\dout(u)=\din(v)=1$ is contracted, and this is repeated until no candidate edges remain.
\end{definition}

For a path $\PP=\bv_1,\bv_2,\dots,\bv_q$ in the original graph, the corresponding path in the condensed graph is obtained by contracting an edge $(\bv_i,\bv_{i+1})$ whenever it is contracted in the graph, replacing the node $\bv_1$ by $\bw$ whenever an edge $(\bu,\bv_1)$ is contracted to form $\bw$, and similarly for the final node $\bv_q$. It is impossible for an intermediate node $\bv_i$, $2\leq i < q$, to be merged with a node outside of $\PP$, as this would violate the condition $\dout(u)=\din(v)=1$ for edge contraction in Defn.~\ref{d:condensed}.

In the condensed sequence graph $\G$ obtained from a sequence $\s$, nodes correspond to subsequences via their labels, and paths in $\G$ correspond to subsequences in $\s$ via merging the constituent nodes along the path. If the subsequence corresponding to a node $\bv$ appears twice or more in $\s$, we say that $\bv$ corresponds to a repeat. Conversely, subsequences of length $\ell \geq K$ in $\s$ correspond to paths $\PP$ of length $\ell - K +1$ in the $K$-mer graph, and thus by the previous paragraph also to paths in the condensed graph $\G$.

We record a few simple facts about the condensed sequence graph obtained from a $K$-mer graph.
\begin{lemma} \label{l:condensed} Let $\G_0$ be the $K$-mer graph constructed from the $(K+1)$-spectrum of $\s$ and let $\CC_0=\CC_0(\s)$ be the cycle corresponding to $\s$. In the condensed graph $\G$, let $\CC$ be the cycle obtained from $\CC_0$ by contracting the same edges as those contracted in  $\G_0$.
\begin{enumerate}
\item Edges in $\G_0$ can be contracted in any order, resulting in the same graph $\G$, so the condensed graph is well-defined. Similarly $\CC$ is well-defined. 
\item  The cycle $\CC$ in $\G$ corresponds to $\s$ and is the unique such cycle. 
\item The cycle $\CC$ in $\G$ traverses each edge at least once.  
\end{enumerate}
\end{lemma}

\begin{reptheorem}{t:SBH_no_multiplicities}
Let $\calS_{K+1}$ be the $(K+1)$-spectrum of $\s$ and $\G_0$ be the $K$-mer graph constructed from $\calS_{K+1}$, and 
let $\G$ be the condensed sequence graph obtained from $\G_0$. If Ukkonen's condition is satisfied, i.e. there are no triple repeats or interleaved repeats of length at least $K$, then there is a unique Eulerian cycle $\CC$ in $\G$ and $\CC$ corresponds to $\s$.
\end{reptheorem}
\begin{proof}
We will show that if Ukkonen's condition is satisfied, the cycle $\CC=\CC(\s)$ in $\G$ corresponding to $\s$ (constructed in Lemma~\ref{l:condensed}) traverses each edge exactly once in the condensed $K$-mer graph, i.e. $\CC$ is Eulerian. Pevzner's \cite{Pev95} arguments show that if there are multiple Eulerian cycles then Ukkonen's condition is violated, so it is sufficient to prove that $\CC$ is Eulerian. As noted in Lemma~\ref{l:condensed}, $\CC$ traverses each edge at least once, and thus it remains only to show that $\CC$ traverses each edge \emph{at most} once.

To begin, let $\CC_0$ be the cycle corresponding to $\s$ in the original $K$-mer graph $\G_0$. 
We argue that every edge $(\bu,\bv)$ traversed twice by $\CC_0$ in the $K$-mer graph $\G_0$ has been contracted in the condensed graph $\G$ and hence in $\CC$. Note that the cycle $\CC_0$ does not traverse any node three times in $\G_0$, for this would imply the existence of a triple repeat of length $K$, violating the hypothesis of the Lemma.  It follows that the node $\bu$ cannot have two outgoing edges in $\G_0$ as $\bu$ would then be traversed three times; similarly, $\bv$ cannot have two incoming edges.  Thus $\dout(\bu)=\din(\bv)=1$ and, as prescribed in Defn.~\ref{d:condensed}, the edge $(\bu,\bv)$ has been contracted. 
\end{proof}

\subsubsection{Proofs for \bdbI}
\label{sec:bdbApp}

Since bridging reads extend one base to either end of a repeat, it will be convenient to use the following notation for extending sequences: Given an {\xn } $\bv$ with an incoming edge $(\bp,\bv)$ and an outgoing edge
 $(\bv,\bq)$, let 
 \begin{equation}\label{e:extNotation}
 \vq=\bv\,\bq_{\olap\bv\bq+1}^{1},\quad \text{and}\quad \pv=\bp_{\vecend - \olap\bp\bv}^{1}\bv\,.
 \end{equation}
Here $\vq$ denotes the subsequence $\bv$ appended with the single next base in the merging of $\bv$ and $\bq$ and $\pv$ the subsequence $\bv$ prepended with the single previous base in the merging of $\bp$ and $\bv$.
 For example, if $\bv =$ ATTC, $\bp=$ TCAT, $a_{\bp\bv}=2$, $\bq = $ TTCGCC, and $a_{\bv\bq} = 3$, then $\vq =$ ATTCG, $\pv =$ CATTC, and $\pvq = $ CATTCG. 
 
 The idea is that a bridging read is consistent with only one pair $\pv$ and $\vq$ and thus allows to match up edge $(\bp,\bv)$ with $(\bv,\bq)$. This is recorded in the following lemma.
 
\begin{lemma}\label{l:bridgingReads}
Suppose $\CC$ corresponds to a sequence $\s$ in a condensed sequence graph $\G$. If a read $\br$ bridges an {\xn } $\bv$, then there are unique edges $(\bp,\bv)$ and $(\bv,\bq)$ such that $\pv$ and $\vq$ are adjacent in $\br$.
\end{lemma}

{\bdbI } is described as follows. 
\begin{algorithm}[H]
\caption{{\sc \bdbI}. Input: reads $\RR$, parameter $K$. Output: sequence $\sh$. }
\label{alg2}
\noindent
\emph{$K$-mer steps 1-3:}\\
1. For each subsequence $\x$ of length $K$ in a read, form a node with label $\x$.  \\
2. For each read, add edges between nodes representing adjacent $K$-mers in the read. \\
3. Condense the graph as described in Defn.~\ref{d:condensed}.\\
4.~\emph{Bridging step:} See Fig.~\ref{fig:bridgingStepEx}. While there exists an {\xn } $\bv$ with $\din(\bv)=\dout(\bv)=2$ bridged by some read $\br$: (i) Remove $\bv$ and edges incident to it. Add duplicate nodes $\bv_1,\bv_2$. (ii) Choose the unique $\bp_i$ and $\bq_j$ s.t. $\pvinp{\bp_i}{\bv}$ and $\vqinp\bv{\bq_j}$ are adjacent in $\br$ and add edges $(\bp_i,\bv_1)$ and $(\bv_1,\bq_j)$. Choose the unused $\bp_i$ and $\bq_j$, add edges $(\bp_i,\bv_2)$ and $(\bv_2,\bq_j)$. (iii) Condense the graph. \\
5.~\emph{Finishing step:} Find an Eulerian cycle in the graph and return the corresponding sequence.
\end{algorithm}

\subsubsection{Proofs for \bdbII}
\label{sec:bdbIIApp}

In this subsection we recall Theorem~\ref{l:dogBones} stating sufficient conditions for correct reconstruction, and derive the corresponding required coverage depth and read length to meet a target probability of correct reconstruction. The subsection concludes with a proof that the sufficient conditions are correct. 

\begin{figure}[htb]
\centering{
\includegraphics[height=5in]{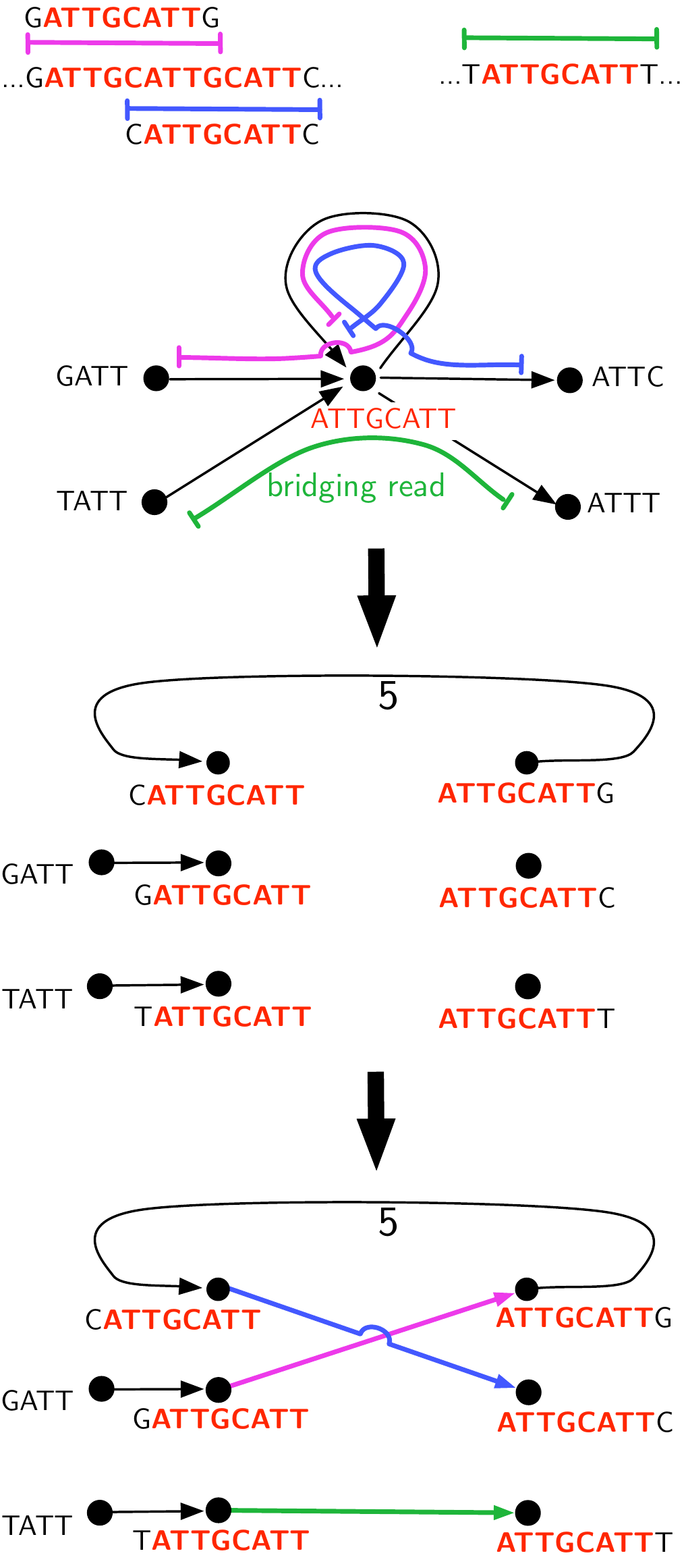}
\caption{Resolution of {\xn } with a self-loop. }
\label{fig:XnodeB} }
\end{figure}

\begin{reptheorem}{l:dogBones}
 The algorithm {\bdbII }  reconstructs the sequence $\s$ if:
\begin{enumerate}
\item[(a)] all interleaved repeats are bridged
\item[(b)] all triple repeats are \textbf{all-bridged}
\item[(c)] the sequence is covered by the reads.
\end{enumerate}
\end{reptheorem}

\begin{remark}Unlike the previous $K$-mer algorithms, \textsc{DeBruijn} and {\bdbI}, it is unnecessary to specify a parameter $K$ for {\bdbII }.  Implicitly {\bdbII }  uses $K=1$, which makes the condition that reads overlap by $K$ equivalent to coverage of the genome. 
\end{remark}

Figure~\ref{fig:plotINTRO} plots the performance of \bdbII, obtained by solving for the relationship between $G,N,L,$ and $\eps$ in order to satisfy the conditions of Lemma~\ref{l:dogBones}. We first perform the requisite calculations, and then prove the Lemma.

Condition~(a) is already dealt with in \eqref{e:LintBound1Appendix} and \eqref{e:LintBound2Appendix}, and Condition~(c) amounts to the requirement that $\frac{N}{\nc}\geq 1$. 


We turn to Condition~(b) that all triple repeats are all-bridged. Let $c_m$ denote the number of triple repeats of length $m$. A union bound estimate over triple repeats for the event that one such triple repeat fails to be all-bridged gives
\begin{equation}
\Pe\approx \sum_m 3\cdot c_m e^{-\lam(L-m-1)^+}\,,
\end{equation} 
and requiring $\Pe\leq \eps$ and solving for $L$ yields
\begin{equation}\label{e:tripBridge}
L\geq \frac1\lam \log\frac{\gam_3}\eps =  \frac G N \log\frac{\gam_3}\eps\,,
\end{equation}
where $\gam_3:= \sum_m 3 c_m e^{(N/G)\cdot (m+1)}$ is computed from the triple repeat statistics $c_m$.

In order to understand the cost of all-bridging triple repeats, compared to simply bridging one copy as required by our lower bound, it is instructive to study the effect of the single longest triple repeat. Setting $c_{\Ltri} = 1$ and $c_m=0$ for $m\neq \Ltri$ makes $\gam_3 =  3e^{(N/G)\cdot ( \Ltri+1)}$ in \eqref{e:tripBridge} and  
\begin{equation}\label{e:allBridgeTriple2}
L\geq L_3^\text{all}:=
\Ltri+1+\frac GN \log 3\eps\inv\,.
\end{equation}
Bridging the longest triple repeat, as shown in Section~\ref{sec:TripleRepeats}, requires
\begin{equation}\label{e:bridgeTriple2}
L\geq L_3:=
\Ltri+1+ \frac G{3N} \log \eps\inv\,.
\end{equation}

Solving for $N$ in equations \eqref{e:bridgeTriple2} and \eqref{e:allBridgeTriple2} gives 
\begin{equation}\label{e:NforTriples}
N_3 \geq \frac G3\cdot \frac{\log\eps\inv}{L-\Ltri-1}
\end{equation}
\begin{equation}
N_3^\text{all}\geq  G\cdot \frac{\log\eps\inv+\log 3}{L-\Ltri-1}\,.
\end{equation}
The ratio is
\begin{equation} \label{e:NratioforTriples}
\frac{N_3^\text{all}}{N_3}= 3\cdot \frac{\log3\eps\inv}{\log \eps\inv}\approx 3.72 \quad \text{for } \eps = 10^{-2}\,.
\end{equation}
This means that if the longest triple repeat is dominant, then for $L$ slightly larger than $\Ltri$, {\bdbII } needs a coverage depth approximately 3.72 times higher than required by our lower bound.

The remainder of this subsection is devoted to the proving Lemma~\ref{l:dogBones}.

We will use $m_{\CC}(\bv)$ to denote the multiplicity (traversal count) a cycle $\CC$ assigns a node $\bv$. The multiplicity $m_{\CC}(\bv)$ is also equal to the number of times the subsequence $\bv$ appears in the sequence corresponding to $\CC$. For an edge $e$, we can similarly let $m_{\CC}(e)$ be the number of times $\CC$ traverses the edge.
The following key lemma relates node multiplicities with the existence of {\xns}.
\begin{lemma}\label{c:SEclaim}
Let $\CC$ be a cycle in a condensed sequence graph $\G$, where $\G$ itself is not a cycle, traversing every edge at least once. If $\bv$ is a node with maximum multiplicity at least 2, i.e. $m_\CC(\bv)=\max_{u\in \G} m_\CC(\bu) \geq 2$, then $\bv$ is an \xn. As a consequence, if $m_\CC(\bv)\geq 3$ for some $\bv$, i.e. $\CC$ traverses some node at least three times, then $m_\CC(\bu)\geq 3$ for some {\xn } $\bu$. 
\end{lemma}
\begin{proof}
Let $\bv$ be a node with maximum multiplicity $m_\CC(\bv)=\max_{\bu\in \G} m_\CC(\bu)$. We will show that $\bv$ is an {\xn }, i.e. $\dout(\bv)\geq 2$ and $\din(\bv)\geq 2$.

We prove that $\dout(\bv)\geq 2$ by supposing that $\dout(\bv)=1$ and deriving a contradiction. Denote the outgoing edge from $\bv$ by $e=(\bv,\bu)$, where $\bu$ is distinct from $\bv$ since otherwise $\G$ is a cycle. If $\din(\bu)\geq 2$, then $\bu$ must be traversed more times than $\bv$, contradicting the maximality of $m_\CC(\bv)$, and if $\din(\bu)=1$, then the existence of the edge $e$ contradicts the fact that $\G$ is condensed. 
The argument showing that $\din(\bv)\geq 2$ is symmetric to the case $\din(\bv)\geq 2$.
\end{proof}


\begin{proof}[Proof of Lemma~\ref{l:dogBones}]
We assume that all triple repeats are all-bridged, that there are no unbridged interleaved repeats, and that all reads overlap their successors by at least $1$ base pair. We wish to show that {\bdbII } returns the original sequence.

Consider the condensed sequence graph $\G_0$ constructed in steps~1-3 of {\bdbII}. 
Suppose all {\xns } that are either all-bridged or correspond to bridged 2-repeats have been resolved according to repeated application of the procedure in step~4 of {\bdbII}, resulting in a condensed sequence graph $\G$. We claim that 1) $\s$ corresponds to a cycle $\CC$ in $\G$ traversing every edge at least once, 2) $\CC$ is Eulerian, and 3) $\CC$ is the unique Eulerian cycle in $\G$.  

\paragraph{Proof of Claim 1.}
Let $\G_n$ be the graph after $n$ resolution steps, and suppose that $\CC_{n}$ is a cycle in $\G_n$ corresponding to the sequence $\s$ and traversing all edges. We will show that there exists a cycle $\CC_{n+1}$ in $\G_{n+1}$ corresponding to $\s$ and traversing all edges, and that $\G_{t}=\G$ for a finite $t$, so by induction, there exists a cycle $\CC$ in $\G$ corresponding to $\s$ and traversing all edges. The base case $n=0$ was shown in Lemma~\ref{l:cycleLemma}. Moving on to arbitrary $n>0$, let $\bv$ be an {\xn } in $\G_n$ labeled as in Fig.~\ref{fig:condensing}. The {\xn } resolution step is constructed precisely to preserve the existence of a cycle corresponding to $\s$. Each traversal of $\bv$ by the cycle $\CC_n$ assigns an incoming edge $(\bp_i\bv)$ to an outgoing edge $(\bv,\bq_j)$, and the resolution step correctly determines this pairing by the assumption on bridging reads. 

Note that all {\xns } in the graph $\G_{n+1}$ continue to correspond to repeats in $\s$. The process terminates: let $\LL(\G_i)=\sum_{\bv\in \G_i} m_{\CC_i}(\bv) \mathbf{1}_{m_{\CC_i}(\bv)>1}$ and observe that $\LL(\G_i)$ is strictly decreasing in $i$. Thus $\s$ corresponds to a cycle $\CC$ in $\G$ traversing each edge at least once. 

\paragraph{Proof of Claim 2.} We next show that $\CC$ is an Eulerian cycle. If $\G$ is itself a cycle, and $\s$ is not formed by concatenating multiple copies of a shorter subsequence (assumed not to be the case, see discussion at end of  Section~\ref{sec:Lower}), then $\CC$ traverses $\G$ exactly once and is an Eulerian cycle. Otherwise, if $\G$ is not a cycle, then we may apply Lemma~\ref{c:SEclaim} to see that any node with $m_{\CC}(\bv)\geq 3$ implies the existence of an {\xn } $\bu$ with $m_{\CC} (\bu)\geq 3$. Node $\bu$ must be all-bridged, by hypothesis, which means that an additional {\xn } resolution step can be applied to $\G$, a contradiction. Thus each node $\bv$ in $\G$ has multiplicity $m_{\CC}(\bv)\leq 2$.
 
We can now argue that no edge $e=(\bu,\bv)$ is traversed twice by $\CC$ in the condensed sequence graph $\G$, as it would have been contracted. Suppose $m_{\CC}(e)\geq 2$. The node $\bu$ cannot have two outgoing edges as this implies $m_\CC(\bu)\geq 3$; similarly, $\bv$ cannot have two incoming edges.  Thus $\dout(\bu)=\din(\bv)=1$, but by Defn.~\ref{d:condensed} the edge $e=(\bu,\bv)$ would have been contracted.


\paragraph{Proof of Claim 3.} It remains to show that there is a \emph{unique}  Eulerian cycle in $\G$. 
All {\xns } in $\G$ must be unbridged {2-\xns} (correspond to 2-repeats in $\s$), as all other {\xns } were assumed to be bridged and have thus been resolved in $\G$. 
 
We will map the sequence $\s$ to another sequence $\s'$, allowing us to use the characterization of Lemma~\ref{l:Pev95} for SBH with known multiplicities. Denote by $\G'$ the graph obtained by relabeling each node in $\G$ by a single unique symbol (no matter the original node label length), and setting all edge overlaps to $0$. Through the relabeling, $\CC$ corresponds to a cycle $\CC'$ in $\G'$, and let $\s'$ be the sequence corresponding to $\CC'$. Writing $\calS_2'$ for the $2$-spectrum of $\s'$,  the graph $\G'$ is by construction precisely the 1-mer graph created from $\calS_2'$, and there is a one-to-one correspondence between {\xns } in $\G'$ and unbridged repeats in $\s'$. Through the described mapping, every unbridged repeat in $\s'$ maps to an unbridged repeat in $\s$, with the \emph{order of repeats preserved}. 

There are multiple Eulerian cycles in $\G$ only if there are multiple Eulerian cycles in $\G'$ since the graphs have the same topology, and by Lemma~\ref{l:Pev95} the latter occurs only if there are unbridged interleaved repeats in $\s'$, which by the correspondence in the previous paragraph implies the existence of unbridged interleaved repeats in  $\s$ . 
\end{proof}

\subsection{Truncation estimate for bridging repeats ({\sc \GR} and {\sc \bdbII})}
\label{sec:greedyBridge}
The repeat statistics $a_m$ and $c_m$ used in the algorithm performance curves are potentially overestimates. This is because a large repeat family---one with a large number of copies $f$---will result in a contribution ${f\choose 2}\approx f^2/2$ to $a_m$ and ${f\choose 3}\approx f^3/6$ to $c_m$. 
  
We focus here on deriving an estimate for the required $N,L$ for bridging all repeats with probability $1-\eps$. This upper bound reduces the sensitivity to large families of short repeats. The analogous derivation for all-bridging all triple-repeats is very similar and is omitted. 
 
Suppose there are $a_m$ repeats of length $m$. The probability that some repeat is unbridged is approximately, by the union bound estimate,  
\begin{equation}
\P(\EE)\approx \sum_m a_m e^{-2\lam(L-m)}\,.
\end{equation} Requiring $\P(\EE)\leq \eps$ and solving for $L$ gives 
\begin{equation}\label{e:greedyBound}
L\geq \frac1{2\lam}\log\frac\gam\eps =\frac G{2N}\log\frac\gam\eps \,,
\end{equation}
where $\gam:=\sum_m a_m e^{2(N/G) m}$.
Now, if $a_m$ \emph{overcounts} the number of repeats for small values of $m$, the bound in \eqref{e:greedyBound} might be loose. 
In order for each read to overlap the subsequent read by $x$ nucleotides, with probability of failure $\eps/2$, it suffices to take
\begin{equation}
L\geq L_\text{K-cov}\Big(x,\frac\eps2\Big):=x + \frac1\lam \log\frac {2N}\eps\,.
\end{equation}
Thus, for any $x< L$, we may replace \eqref{e:greedyBound} by 
\begin{equation}
L\geq \min_x\max\{\frac1{2\lam}\log\frac{2\gam(x)}\eps,L_\text{K-cov}(x,\frac\eps2)\}\,,
\end{equation}
where $\gam(x) = \sum_{m> x} a_m e^{2(N/G) m}$, and obtain a \emph{looser} bound. 

\section{Critical window calculations}

\subsection{Window size if $\Lint \gg \Ltri$}
\label{sec:criticalWindow1}

We focus here on the bound due to interleaved repeats (rather than triple repeats, treated subsequently), and furthermore assume that the effect of the single largest interleaved repeat is dominant. 
In this case $\Lint = \lcrit-1$ is the length of the shorter of the pair of interleaved repeats, and let $\ell_1$ be the length of the longer of the two. For $\lcrit<L\leq \ell_1+1$, we are in the setting of \eqref{e:LintBound1Appendix} but with a redefined $\gamma_1 = e^{2(N/G)(\lcrit-1)}$. Thus, 
\begin{equation}
L\geq\lcrit+ \frac G{2N}\log\eps\inv\,,
\end{equation}
and solving for $N$ gives 
\begin{equation}\label{e:largestInt}
N_\text{repeat} = \frac G2 \frac{\log \eps\inv}{L-\ell_2-1}
\end{equation}
Let $L^*$ be the value of $L$ at which the curve described by constraint \eqref{e:largestInt} intersects the Lander-Waterman coverage value, i.e. $N_\text{repeat}(L^*) = \nc (L^*):=N^*$. This is the minimum read length for which coverage of the sequence suffices for reconstruction.

We now solve for $\frac{L^*}{\lcrit}$. First, the Lander-Waterman equation \eqref{e:cov} at $N=N^*$ is 
\begin{equation}\label{e:cov2}
N^* = \frac{G}{L^*}\log \frac{N^*}\eps\,,
\end{equation}
and setting equal the right-hand sides of \eqref{e:cov2} and \eqref{e:largestInt} at $L=L^*$ gives
$$
 \frac{G}{L^*}\log \frac{N^*}\eps =  \frac G2 \frac{\log \eps\inv}{L^*-\ell_2-1}\,.
$$ 
A bit of algebra yields
\begin{equation}\label{e:CritWidth}
\frac{L^*}{\lcrit} = \frac{2}{2-x}\,,
\end{equation}
where 
\begin{equation}\label{e:windowX}
x:=\cdot \frac{\log\eps\inv}{\log N^*+\log \eps\inv}\,.
\end{equation}
Since $x\leq \frac12$, equation \eqref{e:CritWidth} implies $L^*\leq 2\lcrit$, and combined with the obvious inequality $L^*\geq \lcrit$, we have $\lcrit \leq L^* \leq 2\lcrit$. Thus 
\begin{equation}
\nc(2\lcrit)\leq N^* \leq \nc(\lcrit)\,,
\end{equation} 
and applying the Lander-Waterman  fixed-point equation \eqref{e:cov} yet again gives
\begin{equation}
\frac{G}{2\lcrit}\log \frac{\nc(2\lcrit)}\eps \leq N^* \leq \frac{G}{\lcrit}\log \frac{\nc(\lcrit)}\eps\,.
\end{equation}
Writing this out gives
\begin{align*}
& \frac{\log\eps\inv}{\log \frac G\lcrit+\log\log \frac{\nc(\lcrit)}\eps  +\log \eps\inv}  \leq x \\&\quad\leq \frac{\log\eps\inv}{\log \frac G\lcrit-1+\log\log \frac{\nc(2\lcrit)}\eps  +\log \eps\inv} \,,
\end{align*}
and this can be relaxed to
\begin{equation}\label{e:xVal}
\begin{split}
&\frac{\log\eps\inv}{\log \frac G\lcrit+\log \eps\inv+\log\log\frac G{\eps\lcrit}} \leq x \\ &\qquad\leq\frac{\log\eps\inv}{\log \frac G\lcrit-1+\log \eps\inv} \,.
\end{split}
\end{equation}
Letting
\begin{equation}\label{e:rWindow}
r := \frac{\log \frac G\lcrit}{\log \eps\inv}\,,
\end{equation}
we have to a very good approximation 
\begin{equation}\label{e:simpleWindow}
\frac{L^*}{\lcrit} \approx \frac{2(r+1)}{2(r+1)-1}\,.
\end{equation}

For $G\sim 10^8$, $\lcrit\sim 1000$, and $\eps = 5\%$, we get $\log \frac G\lcrit\approx 13.8$ and $\log \eps\inv\approx 3.0$, so $r\approx 4.6$ and 
$$
\frac{L^*}{\lcrit} = \frac{2(r+1)}{2(r+1)-1} \approx 1.1\,.
$$
From \eqref{e:rWindow} we see that the relative size of $\log\eps\inv$ and $\log \frac G\lcrit$ determines the size of the critical window. If in the previous example $\eps=10^{-5}$, say, then $\frac{L^*}{\lcrit}$ increases to $1.3$. As $\eps$ tends to zero, $r$ approaches zero as well and $\frac{L^*}{\lcrit} \to 2$.


\subsection{Window size if $ \Ltri \gg\Lint$}
We now suppose the single longest triple repeat dominates the lower bound and estimate the size of the critical window. In this case $\Ltri= \lcrit-1$ is the length of the longest triple repeat. Since we don't have matching lower and upper bounds for triple repeats, we separately compute the critical window size for each. 

We start with the lower bound. For $L>\lcrit$, the minimum value of $N$ required in order to bridge the longest triple repeat is given by  \eqref{e:NforTriples} and repeated here:
\begin{equation}\label{e:largestTri}
N_\text{triples} = \frac G3\cdot \frac{\log\eps\inv}{L-\lcrit}\,.
\end{equation}
As for the interleaved repeats case considered earlier, we let $L^*$
be the value of $L$ at which the curve described by constraint \eqref{e:largestTri} intersects the Lander-Waterman coverage value, i.e. $N_\text{triple}(L^*) = \nc (L^*):=N^*$. This is the minimum read length for which coverage of the sequence suffices for reconstruction.

A similar procedure as leading to \eqref{e:CritWidth} gives
$L^*/\lcrit = 3/(3-x)$
with $x$ defined in \eqref{e:windowX}.
One can check that the estimates on $x$ in \eqref{e:xVal} continue to hold, 
and we therefore get
\begin{equation}\label{e:simpleTripleWindow}
\frac{L^*}{\lcrit} \approx \frac{3(r+1)}{3(r+1)-1}\,.
\end{equation}
For the same example as before, $G\sim 10^8$, $\lcrit\sim 1000$, and $\eps = 5\%$, we get $r\approx 4.6$ and 
$$
\frac{L^*}{\lcrit} = \frac{3(r+1)}{3(r+1)-1} \approx 1.06\,.
$$
Changing $\eps$ to $10^{-5}$ makes $\frac{L^*}{\lcrit} \approx 1.17$, and as $\eps$ (and hence also $r$) tends to zero, $\frac{L^*}{\lcrit} \to \frac32$.

The analogous computation for $L^*/\lcrit$ for the upper bound, as given by $N^\text{all}_3$ in \eqref{e:NforTriples}, yields 
\begin{equation}\label{e:tripleWindowUpper}
\frac{L^*}{\lcrit} = \frac{r+1}{r+\frac{\log 3}{\log\eps\inv}} \approx 1.12\,,
\end{equation}
for the example with $G\sim 10^8$, $\lcrit\sim 1000$, and $\eps = 5\%$. The critical window size of the upper bound is about twice as large as that of the lower bound for typical values of $G$ and $\lcrit$, with $\eps$ moderate. But as $\eps\to 0$, we see from \eqref{e:tripleWindowUpper} that $L^*/\lcrit \to \infty$, markedly different to the $L^*/\lcrit \to \frac32$ observed for the lower bound.

\end{document}